\documentclass[envcountsame,envcountsect]{llncs}
\usepackage{amsmath}
\usepackage{graphicx}
\usepackage{amssymb}
\usepackage{listings}
\usepackage{upgreek}
\usepackage{tikz}
\usetikzlibrary{arrows,automata}
\DeclareMathOperator\dif{d\!}

\newif\iftechreport
\techreportfalse
\iftechreport  \fi

\title{Distribution-based Bisimulation for Labelled Markov Processes}
\author{Pengfei Yang\inst{1,2} \and David N. Jansen\inst{1} \and Lijun Zhang\inst{1,2}}
\institute{
State Key Laboratory of Computer Science, Institute of Software, CAS \and
University of Chinese Academy of Sciences \\
	\email{yangpf@ios.ac.cn}, \email{dnjansen@ios.ac.cn}, \email{zhanglj@ios.ac.cn}}

\begin{document}

\maketitle

\begin{abstract}
In this paper we propose a (sub)distribution-based bisimulation for labelled Markov processes and compare it with earlier definitions of state and event bisimulation, which both only compare states.
In contrast to those state-based bisimulations, our distribution bisimulation is weaker, but corresponds more closely to linear properties.
We construct a logic and a metric to describe our distribution bisimulation and discuss linearity, continuity and compositional properties.
\end{abstract}

\section{Introduction}
\subsection{Labelled Markov Processes}
Markov processes are one of the most popular types of stochastic processes in the fields of mathematics, physics, biology, economics, and computer science.
Markov processes have a common property, called \emph{Markov property:}
Given exact information on the present, the future is independent of the past.
There are many examples of Markov processes, like Brownian motion, spread of infectious diseases, option pricing, and quantitative information flow.
In some of these, the state space is continuous,
so it is worth studying such Markov processes.

Labelled Markov processes (LMPs) were first studied in \cite{BDEP97} and \cite{DEP02}.
Contrary to common Markov processes, they contain action labels on the transitions:
There is a set of actions, and for each action there is exactly one subprobabilistic transition function to describe the transition with this action.
That is to say, labelled Markov processes are transition systems with action labels and (sub)probabilistic transitions.
They are input-enabled w.\,r.\,t.\@ fully probabilistic transitions.
We adapt the following example from \cite{AKLP10} to show what is an LMP.

\begin{example}
There are $n$ rooms in a building, and each room has a heater that is either \textbf{On} or \textbf{Off}.
The state space is the state of the heaters and the temperatures of every room,
i.\,e.\@ $S=\{\textbf{On},\textbf{Off}\}^n \times \mathbb{R}^n$.
On every transition we can change the states of heaters, so the set of actions $\mathcal{A}=2^{\{1,2,\ldots,n\}}$.
The temperature of the $i$-th room at time $k$ is denoted by $x_i^{k}$,
and these $x_i$ are determined by the following stochastic difference equation:
\[
x_i^{k+1}=x_i^{k}+b_i(x_0-x_i^{k})+\sum_{j \ne i} a_{ij}(x_j^{k}-x_i^{k})
+c_i\mathbb{I}_{\{q_i^{k}=\textbf{On}\}}+w_i^{k}.
\]
Here $x_0$ is the outside temperature,
$b_i$ is the rate of heat transfer between the $i$-th room and the outside environment,
$a_{ij}$ is the rate of heat transfer from the $j$-th room to the $i$-th room.
$q_i(k)=\textbf{On}$ means from time $k$ to $k+1$ the heater of the $i$-th room is \textbf{On},
$c_i$ describes the temperature influence of this heater,
and $w_i(k)$ are independent normal distribution random variables which represent errors.
Now the state space is no longer discrete, but hybrid,
and we have a discrete-time evolution.
At every step we choose an action from the set $\mathcal{A}$, and the probabilistic transition is determined by a system of difference equations.
\end{example}

\subsection{Related Work and Motivation}
Bisimulation is a useful concept in computer science, especially in formal methods.
It can help us simplify the models and grasp the core properties of systems.
Bisimulation was first studied in \cite{LS91} and \cite{KS60} for discrete probabilistic systems.
On the downside, bisimulations are known to be not \emph{robust:}
a small perturbation of the probabilities may change bisimilar states to become different.
As a result, metrics for probabilistic systems have been proposed,
such that a smaller distance between two states implies their behaviours are more similar.
A distance of zero agrees with the standard (precise) bisimulation.
We refer to \cite[Chapter 8]{P09} for a detailed discussion.
In \cite{BreugelSW08,BBLM13,DHKP16,TangB16},  decision algorithms and optimisations for bisimulation metrics have been investigated.
Bisimulation distance between probabilistic processes composed by standard process is characterised in \cite{GeblerLT15}.
In \cite{A13}, approximating bisimulation based on relations, metrics, and approximating functions for LMPs were discussed systematically.

Bisimulations for Markov processes with continuous state spaces (especially analytic spaces) were studied in \cite{BDEP97}, \cite{DEP98} and \cite{DEP02}.
These papers also introduced the name ``labelled Markov processes''.
They defined bisimulation for LMPs in a coalgebraic way
and constructed a simple logic to characterise this bisimulation.
This work led to a lot of further research on bisimulations for LMPs~\cite{P09}.

Metrics, approximations and other topics based on bisimulation for labelled Markov processes were studied in \cite{DDLP06}, \cite{DGJP04}, \cite{DGJP00}, \cite{DDP04} and \cite{DDP03}.
In \cite{CDPP09}, a bisimulation relation was defined in a categoric way for abstract Markov processes,
and this paper also discussed logical characterisation and approximation based on their bisimulation.
\cite{DTW12} discussed state and event bisimulation for non-deterministic LMPs
and gave a logic characterisation of event bisimulation.

The work mentioned above all focuses on bisimulations between states.
That is to say, their bisimulations are binary relations on the state space.
Inspired by \cite{DoyenHR08},
research on bisimulations based on \emph{distributions} (or \emph{subdistributions}) for probabilistic systems with discrete state spaces bloomed up~\cite{Hennessy12,FZ14,HermannsKK14}.

Distribution-based bisimulations are usually coarser than state-based bisimulations,
i.\,e.\@ they declare more states in probabilistic systems equivalent.
We are not aware of any research on distribution-based bisimulation for LMPs or other probabilistic systems with continuous state spaces or time evolution,
which motivates us to carry on with such research.
There are many methods and results which are inspired by the discrete situation,
but also some new problems, observations and differences have appeared.

Different from state-based bisimulation, distribution-based bisimulation has a tight connection with linear-time properties.
In \cite{FZ14}, an equivalence metric is put forward to measure the distance between two systems.
Basically the metric characterising bisimulation is equal to this equivalence metric, so their distribution-based bisimulation corresponds to trace distribution equivalence.
In our setting, similar results hold,
which indicates that our distribution-based bisimulation characterises equivalence of linear properties.
When discussing distribution bisimulation,
we can construct a logical characterization even for state spaces that are not analytic.
Also, some proofs which are trivial for discrete models need a second thought.

Summarising, the main contributions of our paper are:
\begin{itemize}
\item	First, we propose a distribution-based bisimulation for LMPs
	(Sect.~\ref{sec:subdistribution-bisimulation}).
	We show that our definition conservatively extends standard state-based and event-based bisimulations in the literature.
\item	Second, We provide a logical characterisation result for our bisimulation based on extensions of the Hennessy--Milner logic
	(Sect.~\ref{sec:logical-characterisation}).
\item	Also, we define a (pseudo)metric between distributions of LMPs with discounting factor $0<c\le 1$
	(Sect.~\ref{sec:metrics-approximation}).
	A distance of $0$ implies our notion of bisimilarity.
	Further, we investigate the notion of equivalence metric, characterising trace equivalence distance,
	and show that our metric matches the trace equivalence distance in a natural manner.
	We study some useful properties and then investigate the compositional properties.
\end{itemize}

\section{Subdistribution Bisimulation}
\label{sec:subdistribution-bisimulation}
We assume that the readers have basic knowledge of measure theory, like measurable spaces, (sub)probability measures, Borel $\sigma$-algebra, and integration of a Borel-measurable function.
In Appendix~\ref{A0} we recall some basic definitions and properties that we will make use of. We refer to \cite{D04} for details.

\subsection{Bisimulations for Labelled Markov Processes}
First we introduce the definition of labelled Markov processes (LMPs) formally \cite{BDEP97,P09}.
We equip an LMP with an initial distribution.

\begin{definition}
A labelled Markov process (LMP) is a tuple $(S,\Sigma,(\tau_a)_{a \in \mathcal{A}},\pi)$, where
\begin{itemize}
\item $(S,\Sigma)$ is a measurable space;
\item $\tau_a:S \times \Sigma \to [0,1]$ is a subprobability transition function indexed with an element $a$ in the set $\mathcal{A}$ of actions, where we assume that $\mathcal{A}$ is countable;
\item $\pi \in Dist(S)$ is the initial distribution.
\end{itemize}
Here $(\tau_a)_{a \in \mathcal{A}}$ induces a relation $\to$ on $S \times \mathcal{A} \times subDist(S)$: $(s,a,\mu) \in {\to}$, also denoted by $s \xrightarrow{a} \mu$, if $\tau_a(s,\cdot)=\mu(\cdot)$. For $\mu,\mu' \in subDist(S)$, we write $\mu \xrightarrow{a} \mu'$, if
\begin{equation*}
\mu'(\cdot)=\int \tau_a(s,\cdot) \mu(\dif s).
\end{equation*}
Moreover, the relation $\to$ can be expanded to $subDist(S) \times \mathcal{A}^* \times subDist(S)$ by:
\begin{itemize}
\item $\mu \xrightarrow{\varepsilon} \mu$, where $\varepsilon$ is the empty word;
\item For $w \in \mathcal{A}^*$ and $a \in \mathcal{A}$, write $\mu \xrightarrow{wa} \mu'$ if there exists $\mu''$ s.\,t.\@ $\mu \xrightarrow{w} \mu'' \xrightarrow{a} \mu'$.
\end{itemize}
\end{definition}

Now we will define subdistribution bisimulation, state bisimulation and event bisimulation for LMPs
so that we can compare these bisimulations.
Subdistribution bisimulation extends the discrete version in \cite{FZ14}.

\begin{definition}\label{defi:subbi}
Let $(S,\Sigma,(\tau_a)_{a \in \mathcal{A}},\pi)$ be an LMP. We say a symmetric relation $R \subseteq subDist(S) \times subDist(S)$ is a (subdistribution) bisimulation relation, if $\mu \mathrel{R} \nu$ implies:
\begin{itemize}
\item $\mu(S)=\nu(S)$;
\item For any $a \in \mathcal{A}$ and $\mu \xrightarrow{a} \mu'$, there exists $\nu \xrightarrow{a} \nu'$, s.\,t.\@ $\mu' \mathrel{R} \nu'$.
\end{itemize}
We say $\mu,\nu \in subDist(S)$ are bisimilar, denoted by $\mu \sim_{\mathrm{d}} \nu$,
if there exists a bisimulation relation $R$, s.\,t.\@ $\mu \mathrel{R} \nu$.
\end{definition}

\paragraph{Remark.} The wording of Def.~\ref{defi:subbi} is classical and can be used for non-deterministic LMPs \cite{DTW12} as well.
Since our LMPs do not contain non-determinism,
the second condition holds
if and only if for any $a \in \mathcal{A}$,
$\mu \xrightarrow{a} \mu'$ and $\nu \xrightarrow{a} \nu'$ implies  $\mu' \mathrel{R} \nu'$.

Like other bisimilarity relations, the relation $\sim_{\mathrm{d}}$ is an equivalence relation, and the proof is classical.

\begin{proposition}\label{transitivity}
The relation $\sim_{\mathrm{d}}$ is an equivalence relation.
\end{proposition}

The following example from \cite{FZ14} shows an LMP with a finite state space, which is classical in discussing bisimulation based on distributions.

\begin{example}\label{subbiex}
\begin{figure}
  \centering
\scalebox{0.75}{
  \begin{tikzpicture}[->,>=stealth,auto,node distance=1.7cm,semithick,scale=1, every node/.style={scale=1}]
	\tikzstyle{blackdot}=[circle,fill=black,minimum size=6pt,inner sep=0pt]
	\tikzstyle{state}=[minimum size=0pt,circle,draw,thick]
	\tikzstyle{stateNframe}=[minimum size=0pt]	

	\node[state](s0){$s_0$};
	\node[state](s1)[below of=s0]{$s_1$};	
	\node[state](s2)[below left of=s1]{$s_2$};	
	\node[state](s3)[below right of=s1]{$s_3$};
   \node[state](t3)[right of=s3]{$t_3$};
   \node[state](t1)[above right of=t3]{$t_1$};
   \node[state](t4)[below right of=t1]{$t_4$};
   \node[state](t5)[right of=t4]{$t_5$};
   \node[state](t2)[above right of=t5]{$t_2$};
   \node[state](t6)[below right of=t2]{$t_6$};
   \node[state](t0)[right of=s0,xshift=4.5cm]{$t_0$};

	\path
			
			  (s0) edge								node[left] {$1$} (s1)
			  (s1) edge								node[above] {$\frac 12$} (s3)
			  (s1) edge								node[above] {$\frac 12$} (s2)
			  (s2) edge [loop below]           node[right] {$1$}   (s2)
            (t0) edge								node[above] {$\frac 12$} (t1)
            (t0) edge								node[above] {$\frac 12$} (t2)
            (t1) edge								node[above] {$\frac 13$} (t3)
            (t1) edge								node[above] {$\frac 23$} (t4)
            (t2) edge								node[above] {$\frac 13$} (t5)
            (t2) edge								node[above] {$\frac 23$} (t6)
            (t3) edge [loop below]           node[right] {$1$}   (t3)
            (t6) edge [loop below]           node[right] {$1$}   (t6)

			;
  \end{tikzpicture}}
\caption{\label{1}An example of subdistribition bisimulation: $\delta_{s_0} \sim_{\mathrm{d}} \delta_{t_0}$.}
\end{figure}
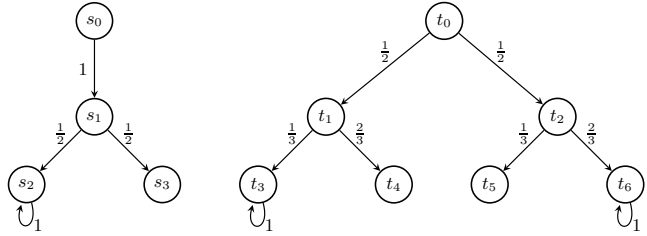
Figure~\ref{1} shows an LMP with a single action in its action set $\mathcal{A}=\{a\}$. In this LMP, we have $\delta_{s_0} \sim_{\mathrm{d}} \delta_{t_0}$. Let the relation $R$ be
$\{(\delta_{s_0}, \linebreak[0] \delta_{t_0}), \linebreak[1]
(\delta_{s_1}, \linebreak[0] \frac 12 \delta_{t_1}+\frac 12 \delta_{t_2}), \linebreak[1]
(\frac 12 \delta_{s_2}+\frac 12 \delta_{s_3}, \linebreak[0] \frac 16 \delta_{t_3}+\frac 13 \delta_{t_4}+\frac 16 \delta_{t_5}+\frac 13 \delta_{t_6}), \linebreak[1]
(\frac 12 \delta_{s_2}, \linebreak[0] \frac 16 \delta_{t_3}+\frac 13 \delta_{t_6})\}.$
Then it is easy to check that its symmetric and reflexive closure $\bar{R}$ is a subdistribution bisimulation relation. Therefore, we have $\delta_{s_0} \sim_{\mathrm{d}} \delta_{t_0}$.
\end{example}

Then we recall state bisimulation according to \cite{DDLP06}.
Given a binary relation $R \subseteq S \times S$, we say $A \subseteq S$ is $R$-closed,
if $R(A):=\{t \in S|\linebreak[0]\exists s \in A, \enspace s \mathrel{R} t\} \subseteq A$.

\begin{definition}\label{defi:statebi}
Let $(S,\Sigma,(\tau_a)_{a \in \mathcal{A}},\pi)$ be an LMP. We say an equivalence relation $R \subseteq S \times S$ is a state bisimulation relation, if $s \mathrel{R} t$ implies that for any $a \in \mathcal{A}$ and $R$-closed set $A \in \Sigma$,
\begin{align}\label{eq:statebi}
\tau_a(s,A)=\tau_a(t,A).
\end{align}
We say $s,t \in S$ are state-bisimilar, denoted by $s \sim_{\mathrm{s}} t$, if there exists a state bisimulation relation $R$, s.\,t.\@ $s \mathrel{R} t$.
\end{definition}

In Def.~\ref{defi:statebi}, we check \eqref{eq:statebi} only for measurable $R$-closed sets.
We do not require all $R$-equivalence classes to be measurable, just as the following example shows.

\begin{example}[\cite{DDLP06}]\label{ex:eventbi}
Let $(\mathbb{R},\mathcal{B}(\mathbb{R}),(\tau_a)_{a \in \{ * \}},\pi)$ be an LMP.
The transitions are defined by $\tau_*(s,\{s\})=1$ for all $s \in \mathbb{R}$.
Let $A \subseteq \mathbb{R}$ be a set which is not Lebesgue-measurable.
Then the relation $R:=(A \times A) \cup (A^c \times A^c)$ is a state bisimulation relation
with non-measurable equivalence classes.
\end{example}

In the example, intuitively we dislike such a bisimulation relation, since the separation is too fine.
To avoid this problem, \cite{DDLP06} defined event bisimulation by:

\begin{definition}
Given a measurable space $(X,\mathcal{F})$, we say $x,y \in X$ is $\mathcal{F}$-in\-distin\-guish\-able, if for any $A \in \mathcal{F}$, either $x,y \in A$ or $x,y \not \in A$.
\end{definition}

\begin{definition}\label{defi:eventbi}
Let $(S,\Sigma,(\tau_a)_{a \in \mathcal{A}},\pi)$ be an LMP.
We say a sub-$\sigma$-algebra $\Lambda \subseteq \Sigma$ is an event bisimulation, if $(S,\Lambda,(\tau_a)_{a \in \mathcal{A}},\pi)$ is still an LMP.
If $\Lambda$ is an event bisimulation, we also say the $\Lambda$-indistinguishable relation, denoted by $\mathcal{R}(\Lambda)$, is an event bisimulation relation.
We say $s,t \in S$ are event-bisimilar, denoted by $s \sim_{\mathrm{e}} t$, if there exists an event bisimulation relation $R$, s.\,t.\@ $s \mathrel{R} t$.
\end{definition}

\subsection{Relations of Bisimulations}
In \cite{DDLP06}, there are several results on the relation between state bisimulation and event bisimulation.
Basically, state bisimilarity always implies event bisimilarity.
For LMPs with analytic spaces as state spaces, event bisimilarity is equivalent to state bisimilarity.
However, for general LMPs, event bisimilarity does not imply state bisimilarity.
See \cite{T10} for a counterexample.

We show that state bisimilarity implies subdistribution bisimilarity.

\begin{theorem}\label{biimply}
Given an LMP with measurable single-point sets.
$s \sim_{\mathrm{s}} t$ implies $\delta_s \sim_{\mathrm{d}} \delta_t$, but $\delta_s \sim_{\mathrm{d}} \delta_t$ does not imply $s \sim_{\mathrm{s}} t$.
\end{theorem}

Consequently, we can extend state bisimilarity $\sim_\mathrm{s}$ to $subDist(S) \times subDist(S)$:
We write $\mu \sim_\mathrm{s} \nu$ if there is a state bisimulation $R$,
s.\,t.\@ for any $R$-closed set $A \in \varSigma$, $\mu(A)=\nu(A)$.
While $\mu \sim_\mathrm{s} \nu$ now implies $\mu \sim_\mathrm{d} \nu$, they are not equivalent.

In \cite{A13}, instead of $R$-closed sets, only equivalence classes are checked in Equ.~\eqref{eq:statebi} of Def.~\ref{defi:statebi}.
However, these two definitions differ, and using equivalence class has counterintuitive consequences.
In particular, Thm.~\ref{biimply} does not hold any more.
The following example shows this fact.

\begin{example}\label{ex1}
Let $\mathcal{M}=([0,1], \mathcal{B}([0,1]), (\tau_a)_{a \in \mathcal{A}},\pi)$ be an LMP, where $\mathcal{A}=\{a\}$ and $\tau_a$ is defined as follows:
\begin{align*}
\tau_a(0,A)=\frac 12 m(A);   \enspace \enspace\enspace \tau_a(1,A)=\frac 12 \int_{A} (x+\frac 12) \dif x; \\
\tau_a(s,\{0\})= \tau_a(s,\{1\})=\frac s2, \enspace \tau_a(s,(0,1))=0,\enspace 0<s<1.
\end{align*}
Here $m$ is the Lebesgue measure on $(\mathbb{R},\mathcal{B}(\mathbb{R}))$.
Let $R \subseteq [0,1] \times [0,1]$ be the smallest equivalence relation that satisfies $0 \mathrel{R} 1$.
Then, the set of equivalence classes contains all singletons $\{ x \}$, for $0 < x < 1$,
and $\{ 0,1 \}$.
It is easy to check that $R$ is not a state bisimulation,
since for the $R$-closed set $I=(0,1/2)$, $\tau_a(0,I) \ne \tau_a(1,I)$.
However, if we replace ``$R$-closed set'' with ``equivalence class'' in Def.~\ref{defi:statebi},
then $R$ is a state bisimulation.

Now we show that $\delta_0 \sim_{\mathrm{d}} \delta_1$ does not hold.
Otherwise,
there exists a bisimulation relation $R'$, s.\,t.\@ $\delta_0 \mathrel{R'} \delta_1$.
Now $\delta_0 \xrightarrow{a} \mu_0$ and $\delta_1 \xrightarrow{a} \mu_1$,
where $\mu_0$ has density $p_0(x)=1/2$ and $\mu_1$ has density $p_1(x)=1/2(x+1/2)$, both on $[0,1]$.
Then we consider the next step $\mu_0 \xrightarrow{a} \mu_0'$ and $\mu_1 \xrightarrow{a} \mu_1'$.
Here we have
$
\mu_0'(\{0\})=\mu_0'(\{1\})=\int_{(0,1)} \frac s 2 ~\mu_0(\dif s)=\int_{(0,1)} \frac 12\cdot \frac s 2 \dif s =\frac 18
$
and
$
\mu_1'(\{0\}) =\mu_1'(\{1\})=\int_{(0,1)} \frac s 2 ~\mu_1(\dif s)
=\int_{(0,1)} \frac 12\left(s+\frac 12\right)\cdot \frac s 2 \dif s =\frac {7}{48}
$.
Because $R'$ is a bisimulation relation, $\mu_0' \mathrel{R'} \mu_1'$, but $\mu_0'(S) \ne \mu_1'(S)$. Contradiction!

Intuitively, the states $0$ and $1$ should not be bisimilar,
since transitions from $0$ and $1$ induce different distributions on $(0,1)$,
where no states appear to be bisimilar.
Therefore, we prefer Def.~\ref{defi:statebi}.

In this example, if we replace $\tau_a(1,\cdot)$ with any non-uniform subdistribution that has measure $1/2$ on $[0,1]$ and mean $1/2$,
then we have $\delta_0 \sim_\mathrm{d} \delta_1$.
However, $0 \sim_\mathrm{s} 1$ still does not hold.
This is a counterexample with a continuous state space showing that subdistribution bisimulation does not imply state bisimulation.
\end{example}

The proof that event bisimulation implies subdistribution bisimulation is more intricate;
we postpone it to the end of the next section.

\section{Logical Characterisation}
\label{sec:logical-characterisation}
Inspired by \cite{BDEP97,DEP98,DEP02,DDLP06} and \cite{FZ14}, we construct a logic to characterise subdistribution bisimulation in this section.
Also, we compare our logical characterisation with that for state bisimulation (\cite{BDEP97}, \cite{DEP98}, \cite{DEP02}) and event bisimulation(\cite{DDLP06}).

\subsection{Logical Characterisation for Subdistribution Bisimulation}

\begin{definition}
We assume a fixed set $\mathcal{A}$ of actions and define a logic given by
\begin{align*}
\mathcal{L}_0::=\mathbb{T}~|~\varphi_1 \wedge \varphi_2~|~\langle a \rangle_q \varphi ~|~ \langle \varepsilon \rangle_q,
\end{align*}
where $a \in \mathcal{A}$ and $q \in \mathbb{Q} \cap [0,1]$, and the formula $\langle \varepsilon \rangle_q$ does not appear in the scope of any diamond operator $\langle a \rangle_q$.
Given an LMP $\mathcal{M}=(S,\Sigma,(\tau_a)_{a \in \mathcal{A}},\pi)$, the semantics are defined inductively as follows:
\begin{itemize}
\item $\mathcal{M},\mu \models \mathbb{T}$,
\item $\mathcal{M},\mu \models \varphi_1 \wedge \varphi_2$ iff $\mathcal{M},\mu \models \varphi_1$ and $\mathcal{M},\mu \models \varphi_2$,
\item $\mathcal{M},\mu \models \langle a \rangle_q \varphi$ iff $\mu'(S)\ge q$ and $\mu' \models \varphi$, where $\mu \xrightarrow{a} \mu'$,
\item $\mathcal{M},\mu \models \langle \varepsilon \rangle_q$ iff $\mu(S) \ge q$.
\end{itemize}
We write $\mathcal{M} \models \varphi$, if $\mathcal{M}, \pi \models \varphi$. If there is no misunderstanding, we simply write $\mu \models \varphi$ instead of $\mathcal{M},\mu \models \varphi$.
\end{definition}

Our formulae $\langle a \rangle_q \varphi$ look similar to the logic defined in \cite{BDEP97},
but their semantics are quite different.
We only care about whether the subdistribution of the next step satisfies $\varphi$
and not about the states any more.
In addition, we have added $\langle \varepsilon \rangle_q$
to measure how ``large'' the subdistribution is,
since subdistribution bisimulation requires that two subdistributions have the same measure on $S$.
If we only consider bisimulation between full distributions, then $\langle \varepsilon \rangle_q$ can be omitted.
Also, we request that $\langle \varepsilon \rangle_q$ does not appear in the scope of any diamond operator $\langle a \rangle_q$ because $\langle a \rangle_0 (\langle \varepsilon \rangle_q \wedge \varphi)$ is semantically equivalent to $\langle a \rangle_q \varphi$, and it is unnecessary to have so many complex formulas.

Now we show that the logic $\mathcal{L}_0$ characterises subdistribution bisimulation.

\begin{theorem}\label{logicalcharacterization}
$\mu \sim_{\mathrm{d}} \nu$ if and only if $\mu$ and $\nu$ satisfy the same formulae in $\mathcal{L}_0$, i.\,e.\@ $\mathcal{L}_0$ characterises subdistribution bisimulation.
\end{theorem}

Next we define four extensions of $\mathcal{L}_0$, which are inspired by \cite{BDEP97}, \cite{DEP98} and \cite{DEP02}.
\begin{align*}
\mathcal{L}_{\mathrm{Can}}&:=\mathcal{L}_0~|~\mathrm{Can}(a),
&
\mathcal{L}_{\neg}&:=\mathcal{L}_0~|~\neg \varphi, \\
\mathcal{L}_{\Delta}&:=\mathcal{L}_0~|~\Delta_a,
&
\mathcal{L}_{\wedge}&:=\mathcal{L}_\neg~|~\bigwedge_{i \in \mathbb{N}} \varphi_i,
\end{align*}
where $a \in \mathcal{A}$. Given an LMP $\mathcal{M}=(S,\Sigma,(\tau_a)_{a \in \mathcal{A}},\pi)$, the semantics are defined inductively as follows:
\begin{itemize}
\item $\mu \models \mathrm{Can}(a)$ iff $\mu'(S)>0$, where $\mu \xrightarrow{a} \mu'$;
\item $\mu \models \Delta_a$ iff $\mu'(S)=0$, where $\mu \xrightarrow{a} \mu'$;
\item $\mu \models \neg \varphi$ iff $\mu \not \models \varphi$;
\item $\mu \models \bigwedge_{i \in \mathbb{N}} \varphi_i$ iff for all $i \in \mathbb{N}$, $\mu \models \varphi_i$.
\end{itemize}

These four extended logics all characterise subdistribution bisimulation.

\begin{proposition}\label{pro:extendedlogic}
$\mathcal{L}_{\mathrm{Can}}$, $\mathcal{L}_{\Delta}$, $\mathcal{L}_{\neg}$ and $\mathcal{L}_{\wedge}$ all characterise subdistribution bisimulation.
\end{proposition}

In previous research of state bisimulation (\cite{BDEP97}, \cite{DEP02}),
only $\mathcal{L}_\wedge$ characterises equivalence classes.
Here we have the following similar result.

\begin{proposition}
$\mathcal{L}_{\wedge}$ characterises bisimilarity equivalence classes,
i.\,e.\@ for any LMP and any equivalence class $C \subseteq subDist(S)$,
there exists a formula $\varphi \in \mathcal{L}_\wedge$,
s.\,t.\@ for any $\mu \in subDist(S)$, $\mu \in C$ if and only if $\mu \models \varphi$.
\end{proposition}
\begin{proof}
Let $C \subseteq subDist(S)$ be a bisimilarity equivalence class. Let $F(C)$ be the set of $\mathcal{L}_0$ formulae which are satisfied by the subdistributions in $C$. It is easy to see that $F(C)$ is countable. Let $\varphi=\bigwedge_{\psi \in F(C)} \psi \in \mathcal{L}_\wedge$. Then for any $\mu \in subDist(S)$, $\mu \in C$ if and only if for any $\psi \in F(C)$, $\mu \models \psi$, i.\,e.\@ $\mu \models \varphi$.
\end{proof}

However, the other logics cannot characterise equivalence classes.

\begin{example}
Let $\mathcal{M}_0$ be an LMP with one action $a$ and only one state $s_0$ going to itself through the action $a$ with probability $1$.
Let $\mathcal{M}_n$ be an LMP with one action $a$ and $n$ states which can do the action $n-1$ times and finally goes to a dead state.
Consider the union LMP $\mathcal{M}=\bigcup_{n=0}^{\infty} \mathcal{M}_n$,
then the equivalence class of $\delta_{s_{0}}$ cannot be characterised by any finite formula.
\end{example}

While for state bisimulation, $\mathcal{L}_\neg$ characterises equivalence classes of any finite LMP,
its subdistribution bisimilarity equivalence classes still cannot be characterised by $\mathcal{L}_0$, $\mathcal{L}_{\mathrm{Can}}$, $\mathcal{L}_{\Delta}$ or $\mathcal{L}_{\neg}$,
as shown by the next example.

\begin{example}
Let $\mathcal{M}$ be an LMP with one action $a$ and two states:
$s$ going to itself with probability $1$,
and $t$ going to itself with probability $0.5$.
We can see that the two states (or rather $\delta_s$ and $\delta_t$) are not bisimilar.
%We can see that on $\mathcal{M}$ the bisimilarity relation ${\sim_{\mathrm{d}}}=\mathrm{ID}$.
Then the equivalence class $\{\sqrt 2 /2 ~\delta_s\}$ cannot be characterised by any finite formula
because an irrational number must be characterised by an infinite sequence of rational numbers.
Moreover, even $\mathcal{L}_0$ ($\mathcal{L}_{\mathrm{Can}}$, $\mathcal{L}_{\Delta}$ or $\mathcal{L}_{\neg}$) cannot characterise equivalence classes of distributions.
Consider the equivalence class $\{\sqrt 2 /2 ~\delta_s+(1-\sqrt 2 /2)~\delta_t\}$:
it is still impossible to characterise an irrational number.
\end{example}

\subsection{Comparison of Logical Characterisations}

In this part we recall the logical characterisation for state-based bisimulation and compare it with ours,
to understand the difference between them deeper.
Also, we will show that event bisimilarity implies subdistribution bisimilarity.
First let's recall the logic that characterises state-based bisimulation (\cite{BDEP97}, \cite{DEP02}).

\begin{definition}
We assume a fixed set $\mathcal{A}$ of actions and define a logic given by
\begin{align*}
\mathcal{L}::=\mathbb{T}~|~\phi_1 \wedge \phi_2~|~\langle a \rangle_q^\mathrm{st} \phi,
\end{align*}
where $a \in \mathcal{A}$ and $q \in \mathbb{Q}$. Given an LMP $\mathcal{M}=(S,\Sigma,(\tau_a)_{a \in \mathcal{A}},\pi)$, the semantics are defined inductively as follows:
\begin{itemize}
\item $\mathcal{M},s \models \mathbb{T}$,
\item $\mathcal{M},s \models \phi_1 \wedge \phi_2$ iff $\mathcal{M},s \models \phi_1$ and $\mathcal{M},s \models \phi_2$,
\item $\mathcal{M},s \models \langle a \rangle_q^\mathrm{st} \phi$ iff there exists $A \in \Sigma$, s.t. $\mu(A) \ge q$, and $t \models \phi$ for all $t \in A$.
\end{itemize}
If there is no misunderstanding, we simply write $s \models \phi$ instead of $\mathcal{M},s \models \phi$.
\end{definition}

The formula $\langle a \rangle_q^\mathrm{st} \phi$ looks similar to $\langle a \rangle_q \varphi$ in $\mathcal{L}_0$.
However, their semantics differ.
For $\langle a \rangle_q^\mathrm{st} \phi$, satisfibility requests a measurable set which is large enough and only contains states satisfying $\phi$,
but for $\langle a \rangle_q \varphi$, we only request that after an action $a$,
the resulting subdistribution should be large enough and satisfy $\varphi$.

From \cite{BDEP97} and \cite{DEP02}, we know that the logic $\mathcal{L}$ can characterise state bisimulation for LMPs with analytic state spaces.
In \cite{DDLP06}, it is proven that $\mathcal{L}$ characterises event bisimulation for arbitrary LMPs.
To conclude, we have the following results:

\begin{proposition}\label{prop:statebilogic}
(1) For an LMP with an analytic state space, $s \sim_\mathrm{s} t$ if and only if $s$ and $t$ satisfy the same formulae in $\mathcal{L}$.

(2) For any LMP, $s \sim_\mathrm{e} t$ if and only if $s$ and $t$ satisfy the same formulae in $\mathcal{L}$.
\end{proposition}

Now we consider whether event bisimilarity implies subdistribution bisimilarity.
We only need to show that, if $s$ and $t$ satisfy the same formulae in $\mathcal{L}$, then $\delta_s$ and $\delta_t$
(provided that every single-point set is measurable)
satisfy the same formulae in $\mathcal{L}_0$.
We note that $\delta_s$ and $\delta_t$ satisfy the same formulae of the form $\langle \varepsilon \rangle_q$,
so we do not consider such formulae any more.
Then the syntaxes of the two logics $\mathcal{L}$ and $\mathcal{L}_0$ become very similar.
We inductively define a mapping $f:\mathcal{L} \to \mathcal{L}_0$ by:
\begin{itemize}
\item $f(\mathbb{T})=\mathbb{T}$,
\item $f(\phi_1 \wedge \phi_2)=f(\phi_1) \wedge f(\phi_2)$,
\item $f(\langle a \rangle_q^\mathrm{st} \phi)=\langle a \rangle_q f(\phi)$.
\end{itemize}
Basically we just replace every $\langle a \rangle_q^\mathrm{st}$ in $\mathcal{L}$ formulae with $\langle a \rangle_q$. Obviously this $f$ is surjective. First we have the following observation:

\begin{proposition}
(1) In $\mathcal{L}_0$, we have $\langle a \rangle_q (\varphi_1 \wedge \varphi_2) \equiv (\langle a \rangle_q \varphi_1) \wedge (\langle a \rangle_q \varphi_2)$, where $\equiv$ means semantic equivalence.

(2) In $\mathcal{L}$, $s \models \langle a \rangle_q^\mathrm{st} (\phi_1 \wedge \phi_2)$ implies $s \models \langle a \rangle_q^\mathrm{st} \phi_1$ and $s \models \langle a \rangle_q^\mathrm{st} \phi_2$.
\end{proposition}

The proposition is easy to prove from the semantics of $\mathcal{L}_0$ and $\mathcal{L}$. From this observation, first we can turn every formula in $\mathcal{L}_0$ to a conjunctive normal form (CNF) $\bigwedge_{i=1}^m \varphi_i$, where every $\varphi_i$ has the form $\langle a_{1} \rangle_{q_{1}} \cdots \langle a_{n_i} \rangle_{q_{n_i}} \mathbb{T}$. First we deal with formulae like $\phi_i$. We have the following proposition:

\begin{proposition}
Given an LMP with measurable single-point sets.
We have that % added to get a nicer line break
$s \models \langle a_1 \rangle_{q_1}^\mathrm{st} \cdots  \langle a_n \rangle_{q_n}^\mathrm{st} \mathbb{T}$ is equivalent to $\delta_s \models \langle a_1 \rangle_{q_1} \cdots \langle a_{n} \rangle_{q_{n}} \mathbb{T}$.
\end{proposition}

For a general formula in $\mathcal{L}_0$, we compare its $f^{-1}$-image
with the $f^{-1}$-image of its CNF.
The latter implies the former,
as transforming a formula in $\mathcal{L}$ to CNF may lead to a weaker formula.
Therefore we get the following result:

\begin{proposition}\label{eventimply}
Given an LMP with measurable single-point sets.
If $s$ and $t$ satisfy the same formulae in $\mathcal{L}$, then $\delta_s$ and $\delta_t$ satisfy the same formulae in $\mathcal{L}_0$.
\end{proposition}

Then from Prop.~\ref{prop:statebilogic}, we immediately get the following result:

\begin{theorem}\label{thm:eventdis}
Given an LMP with measurable single-point sets.
$s \sim_\mathrm{e} t$ implies $\delta_s \sim_\mathrm{d} \delta_t$, but the other direction does not hold.
\end{theorem}

\section{Metrics}
\label{sec:metrics-approximation}
In this section we will introduce a pseudometric
and an approximating subdistribution bisimulation.
% The following sentences are not really helpful.
%, and with the help of the metric we define,
%we can have many useful results of our subdistribution bisimulation.

Given a nonempty set $X$, we say a function $d:X \times X \to [0,\infty)$ is a pseudometric on $X$,
if for all $x,y,z \in X$, we have $d(x,x)= 0$,
symmetry $d(x,y)=d(y,x)$,
and the triangle inequality $d(x,y)+d(y,z) \ge d(x,z)$.
If in addition $d(x,y)=0$ always implies $x=y$, then $d$ is a metric.

\subsection{Metrics and Approximating Bisimulation}
First we give the definition of the pseudometric $d^c$, which is inspired by \cite{DGJP04}.

\begin{definition}\label{defi:metric}
Let $(S,\Sigma,(\tau_a)_{a \in \mathcal{A}},\pi)$ be an LMP. We define $d^c:subDist(S) \times subDist(S) \to [0,1]$ as follows:
\begin{align*}
d^c(\mu,\nu):=\sup_{w \in \mathcal{A}^*,\mu \xrightarrow{w} \mu',\nu \xrightarrow{w} \nu'} c^{|w|}|\mu'(S)-\nu'(S)|,
\end{align*}
where $c \in (0,1]$ is a constant called the discounting factor, and $|w|$ is the length of the word $w$.
\end{definition}

It is obvious that $d^c$ is indeed a pseudometric.
Although $d^c$ is not a proper metric since different subdistributions may have distance $0$,
we follow earlier papers and call this $d^c$ a metric.

Then, the (pseudo)metric $d^c$ characterises subdistribution bisimulation.

\begin{theorem}\label{thm:metric}
(1) $\mu \sim \nu$ implies that for any $c \in (0,1]$, $d^c(\mu,\nu)=0$;

(2) $\mu \sim \nu$ if there exists $c \in (0,1]$, s.\,t.\@ $d^c(\mu,\nu)=0$.
\end{theorem}

With a metric $d^c$ characterising subdistribution bisimulation, we can define approximating bisimilarity through this metric.

\begin{definition}
Let $(S,\Sigma,(\tau_a)_{a \in \mathcal{A}},\pi)$ be an LMP. Given $\epsilon \ge 0$ and $c \in (0,1]$, we say $\mu,\nu \in subDist(S)$ are $\epsilon$-bisimilar with the discounting factor $c$, denoted by $\mu \sim_\epsilon^c \nu$, if $d^c(\mu,\nu) \le \epsilon$.
\end{definition}

It is easy to prove the following properties of approximating bisimilarity.

\begin{proposition}
(1) For any $c \in (0,1]$, ${\sim_\mathrm{d}}={\sim_0^c}$;

(2) For any $c \in (0,1]$ and $0 \le \epsilon \le \epsilon'$, ${\sim_{\epsilon}^c} \subseteq{\sim_{\epsilon'}^c}$;

(3) For any $c \in (0,1]$, ${\sim_\mathrm{d}}=\bigcap_{\epsilon>0} {\sim_{\epsilon}^c}$;

(4) For any $\epsilon \ge 0$ and $0 < c \le c' \le 1$, ${\sim_{\epsilon}^c} \subseteq{\sim_{\epsilon}^{c'}}$;

(5) If $\mu_1 \sim_{\epsilon}^c \mu_2$ and $\mu_2 \sim_{\epsilon'}^c \mu_3$, then $\mu_1 \sim_{\epsilon+\epsilon'}^c \mu_3$.
\end{proposition}

Different from other papers (\cite{FZ14}, \cite{A13}), we directly define our approximating bisimilarity based on the metric, not on an approximating bisimulation relation.
In fact, we could also do the latter, and the two definitions are equivalent:

\begin{definition}
Given a discounting factor $c \in (0,1]$, we say a collection of symmetric relations $\{R_\epsilon^c\}_{\epsilon > 0}$ on $subDist(S)$ is an approximating bisimulation relation with the discounting factor $c$, if $\mu \mathrel{R_\epsilon^c} \nu$ implies:
\begin{itemize}
\item $|\mu(S)-\nu(S)| \le \epsilon$;
\item For any $a \in \mathcal{A}$ and $\mu \xrightarrow{a} \mu'$, there exists $\nu \xrightarrow{a} \nu'$, s.\,t.\@ $\mu' \mathrel{R_{\epsilon/c}^c} \nu'$.
\end{itemize}
We write $\mu \approx_{\epsilon'}^c \nu$, if there exists an approximating bisimulation relation $\{R_\epsilon^c\}_{\epsilon > 0}$, s.\,t.\@ $\mu \mathrel{R_{\epsilon'}^c} \nu$.
\end{definition}

Then we have the following property:
 \begin{proposition}\label{prop:approx}
 ${\sim_{\epsilon'}^c}={\approx_{\epsilon'}^c}$ for any $c \in (0,1]$ and $\epsilon'>0$.
\end{proposition}

\subsection{Equivalence Metric for LMP}
In \cite{FZ14}, distribution-based bisimulation for probabilistic automata \cite{SegalaL95} is constructed,
and an equivalence metric to describe linear-time properties is defined.
Basically, their equivalence metric is the supremum of the distribution difference on finite words.
In probabilistic automata, every state is labelled with a set of atomic propositions.
Not so in LMPs;
however, we can label every state in an LMP with the same label $\top$,
with the intuitive meaning: the process does not stop or block;
then, distribution on traces are just the same as distributions on paths.
Then we can define trace equivalence for two subdistributions in an LMP as follows:
Given an LMP $\mathcal{M}=(S,\Sigma,(\tau_a)_{a \in \mathcal{A}},\pi)$, we say $\pi_1,\pi_2 \in subDist(S)$ are trace equivalent,
if for any $w \in \mathcal{A}^*$, $\pi_1(w)=\pi_2(w)$, where $\pi(w)=\mu(S)$, provided $\pi \stackrel{w}{\to} \mu$.
Also we can define equivalence metric for LMPs as follows:

\begin{definition}[Equivalence Metric]
Let $\mathcal{M}_i=(S_i,\Sigma_i,(\tau_a^i)_{a \in \mathcal{A}},\pi_i)$ for $i=1,2$ be two LMPs. We say $\mathcal{M}_1$ and $\mathcal{M}_2$ are $\epsilon$-equivalent, denoted by $\mathcal{M}_1 \sim_\epsilon \mathcal{M}_2$, if for any word $w \in \mathcal{A}^*$, $|\pi_1(w)-\pi_2(w)| \le \epsilon$. The equivalence metric between $\mathcal{M}_1$ and $\mathcal{M}_2$ is defined by $D(\mathcal{M}_1,\mathcal{M}_2)=\inf\{\epsilon \ge 0 : \mathcal{M}_1 \sim_\epsilon \mathcal{M}_2\}$.
\end{definition}

From the definition, it is obvious
that this metric $D$ is equivalent to our metric $d^1$.
From Prop.~\ref{prop:approx}, it also corresponds
to our approximating bisimulation relation $\sim_\epsilon^1$.
%From this point of view, our approximating bisimulation also describes the distance between two LMPs in the sight of linear properties.
Therefore, we claim that our approximating bisimulation describes the distance between two LMPs with respect to linear properties.
Also, subdistribution bisimilarity is equivalent to trace equivalence.

In some papers (\cite{DGJP04,FZ14}), metrics are defined through a logic. Here in a similar way we can define a metric $d_{\mathrm{l}}^c$ based on a logic. Furthermore, we will show that this metric is equivalent to $d^c$.

\begin{definition}\label{defi:LMc}
Let $c \in (0,1]$ be a discounting factor. Let $\mathcal{M}=(S,\Sigma,(\tau_a)_{a \in \mathcal{A}},\pi)$ be an LMP. We define a logic given by
\begin{align*}
\mathcal{L}_\mathcal{M}^c::= \mathbf{1}~|~\varphi \oplus p~|~ \neg \varphi~|~\bigwedge_{i \in I} \varphi_i~|~ \langle a \rangle^c \varphi,
\end{align*}
where $p \in [0,1]$, $a \in \mathcal{A}$ and $I$ is an index set.
The semantics of the formula $\varphi$ in $\mathcal{L}_\mathcal{M}^c$ is a function on $subDist(S)$, defined inductively as follows:
\begin{align*}
\mathbf{1}(\mu) & :=\mu(S) \\
(\varphi \oplus p)(\mu) & := \min\{\varphi(\mu)+p,1\} \\
\neg \varphi(\mu) & :=1-\varphi(\mu) \\
(\bigwedge_{i \in I} \varphi_i)(\mu) & :=\inf \{\varphi_i(\mu):i \in I\} \\
\langle a\rangle^c \varphi (\mu) & :=c \cdot \varphi(\mu')\text{, where } \mu \xrightarrow{a} \mu'.
\end{align*}
\end{definition}

\begin{definition}
Let $\mathcal{M}=(S,\Sigma,(\tau_a)_{a \in \mathcal{A}},\pi)$ be an LMP. For $\mu,\nu \in subDist(S)$, we define $d_{\mathrm{l}}^c: subDist(S) \times subDist(S) \to [0,1]$ as follows:
\begin{align*}
d_{\mathrm{l}}^c(\mu,\nu):=\sup_{\varphi \in \mathcal{L}_\mathcal{M}^c} |\varphi(\mu)-\varphi(\nu)|.
\end{align*}
\end{definition}

Obviously $d_{\mathrm{l}}^c$ is indeed a pseudometric.
The next theorem shows that $d_{\mathrm{l}}^c$ defined through logic is equivalent to $d^c$.

\begin{proposition}\label{prop:logicalmetric}
Let $\mathcal{M}=(S,\Sigma,(\tau_a)_{a \in \mathcal{A}},\pi)$ be an LMP. Then for any $\mu,\nu \in subDist(S)$ and $c \in (0,1]$, $d^c(\mu,\nu)=d_{\mathrm{l}}^c(\mu,\nu)$.
\end{proposition}

\subsection{Linearity and Continuity of Subdistribution Bisimulation}
Theorem~\ref{thm:metric} is powerful, because with it we can prove some properties of the relation $\sim_{\mathrm{d}}$ more easily.
In this part, we illustrate how to exploit them to prove the linearity and continuity of our subdistribution bisimulation.
In \cite{FZ14,Deng}, similar results have been proven for discrete models.
However, for LMPs with arbitrary state spaces, the proofs are quite different.
Here approximation with simple functions and the monotone convergence theorem are applied multiple times, which indicates the intuition
that it is a good way to use finite models to approximate an LMP in many problems.

Given a sequence of subdistributions $\{\mu_n\}$ on $(X,\mathcal{F})$, we say $\{\mu_n\}$ converges to $\mu$, denoted by $\mu_n \to \mu$, or $\lim_{n\to\infty} \mu_n =\mu$, if for any $A \in \mathcal{F}$, $\mu_n(A) \to \mu(A)$ as $n \to \infty$. Now we give the definitions of linearity, $\sigma$-linearity and continuity of a relation on $subDist(S)$.

\begin{definition}
We say a relation $R\subseteq subDist(S) \times subDist(S)$ is linear,
if for any $\mu_i \mathrel{R} \nu_i$, $i=1,2,\ldots,n$, and $\{a_i\}_{i=1}^n$ s.\,t.\@ $\sum_{i=1}^n a_i\mu_i$ as well as $\sum_{i=1}^n a_i\nu_i$ are subdistributions, where $a_i \ge 0$,
we have $\sum_{i=1}^n a_i\mu_i \mathrel{R} \sum_{i=1}^n a_i\nu_i$.

We say a relation $R\subseteq subDist(S) \times subDist(S)$ is $\sigma$-linear,
if for any $\mu_i \mathrel{R} \nu_i$, $i=1,2,\ldots$, and $\{a_i\}_{i=1}^\infty$ s.\,t.\@ $\sum_{i=1}^\infty a_i\mu_i$ as well as $\sum_{i=1}^\infty a_i\nu_i$ are subdistributions, where $a_i \ge 0$, we have $\sum_{i=1}^\infty a_i\mu_i \mathrel{R} \sum_{i=1}^\infty a_i\nu_i$.

We say a relation $R\subseteq subDist(S) \times subDist(S)$ is continuous, if for any $\mu_i \mathrel{R} \nu_i$, $i=1,2,\ldots$, with $\mu_i \to \mu$ and $\nu_i \to \nu$ as $n \to \infty$, we have $\mu \mathrel{R} \nu$.
\end{definition}

We first discuss linearity and $\sigma$-linearity.
We need a lemma showing that the relation $\mathord{\xrightarrow{w}} \subseteq subDist(S) \times subDist(S)$ is linear and $\sigma$-linear
on the space of Borel-measurable functions.

\begin{lemma}\label{lemma:linear}
For any $w \in \mathcal{A}^*$, the relation $\xrightarrow{w}$ is linear and $\sigma$-linear.
\end{lemma}

Then we have the following linear and $\sigma$-linear properties.

\begin{proposition}
The relation $\sim_{\mathrm{d}}$ is linear and $\sigma$-linear.
\end{proposition}
\begin{proof}
We assume $\mu_i \sim_{\mathrm{d}} \nu_i$, and we have $d^c(\mu_i,\nu_i)=0$, i.\,e.\@ for any $w \in \mathcal{A}^*$, $\mu_i'(S)=\nu_i'(S)$, where $\mu_i \xrightarrow{w} \mu_i'$ and $\nu_i \xrightarrow{w} \nu_i'$. Then from the linearity of $\xrightarrow{w}$, we have $\sum_{i=1}^n a_i\mu_i \xrightarrow{w} \sum_{i=1}^n a_i\mu_i'$ and $\sum_{i=1}^n a_i\nu_i \xrightarrow{w} \sum_{i=1}^n a_i\nu_i'$, and naturally
\begin{align*}
\sum_{i=1}^n a_i\mu_i'(S)=\sum_{i=1}^n a_i\nu_i'(S),
\end{align*}
which indicates $d^c(\sum_{i=1}^n a_i\mu_i,\sum_{i=1}^n a_i\nu_i)=0$ since $w$ is arbitrary.

By taking the limit $n \to \infty$ in the proof above, we can see that the relation $\sim_{\mathrm{d}}$ also is $\sigma$-linear.
\end{proof}

Now we discuss continuity.
Similarly we only need to prove that the relation $\xrightarrow{w}$ is continuous.

\begin{lemma}\label{lemma:continuous}
The relation $\xrightarrow{w}$ is continuous.
\end{lemma}

Actually from the proof of Lemma~\ref{lemma:continuous}, we can get a stronger result:
If $\mu_i \xrightarrow{w} \nu_i$ and $\lim_{i \to \infty} \mu_i = \mu$,
then there exists a subdistribution $\nu$, s.\,t.\@ $\lim_{i \to \infty} \nu_i = \nu$.
Then it is natural that the relation $\sim_\mathrm{d}$ is continuous.

\begin{proposition}\label{prop:continuity}
The relation $\sim_\mathrm{d}$ is continuous.
\end{proposition}
\begin{proof}
We assume $\mu_i \sim_{\mathrm{d}} \nu_i$, $\mu_i \to \mu$ and $\nu_i \to \nu$.
We need to prove $\mu\sim_\mathrm{d}\nu$.
From Thm.~\ref{thm:metric} we have $d^c(\mu_i,\nu_i)=0$,
i.\,e.\@ for any $w \in \mathcal{A}^*$, $\mu_i'(S)=\nu_i'(S)$,
where $\mu_i \xrightarrow{w} \mu_i'$ and $\nu_i \xrightarrow{w} \nu_i'$.
Because $\mu_i \to \mu$ and $\nu_i \to \nu$,
there exist $\mu'$ and $\nu'$, s.\,t.\@ $\lim_{i \to \infty} \mu_i' = \mu'$ and $\lim_{i \to \infty} \nu_i' = \nu'$,
and we have $\mu \xrightarrow{w} \mu'$ and $\nu \xrightarrow{w} \nu'$.
Then
$\mu'(S)=\lim_{i \to \infty} \mu_i'(S)=\lim_{i \to \infty} \mu_i'(S)=\nu'(S)$,
which indicates $d^c(\mu,\nu)=0$ since $w$ was arbitrary.
\end{proof}

\begin{example}
Let $\mathcal{M}=(S,\Sigma,(\tau_a)_{a \in \mathcal{A}},\pi)$ be an LMP, where $S=[0,1]$, $\Sigma=\mathcal{B}([0,1])$, and $\mathcal{A}=\mathbb{N}$. We use $E_0$ to denote the set $[0,1]\setminus(1/3,2/3)$, and $E_n$ to denote the set obtained by removing the middle third of each interval that remains in $E_{n-1}$. The limit set $C=\lim_{n \to \infty} E_n$ is called the Cantor set. (See \cite{D04} for more details.) We define the transitions as follows:
\begin{equation*}
\tau_i(x,A)=\begin{cases}
\frac 12(\frac 32)^{i+1} m(A \cap E_i^c), & \text{if } x \in E_i, \\
\delta_x(A), & \text{otherwise},
\end{cases}
\end{equation*}
where $m$ is the Lebesgue measure on $([0,1],\mathcal{B}([0,1]))$. First, it is easy to see that this $\mathcal{M}$ is indeed an LMP. We use $U(E)$ to denote the uniform distribution over $E \in \mathcal{B}([0,1])$ with $m(E)>0$. We can see that $U([0,3^{-n-1}])$ and $U(E_n)$ are subdistribution bisimilar because these two distributions have the same subdistribution after any $a_i$-transition. It is obvious that $U([0,3^{-n-1}])$ converges to the Dirac distribution $\delta_0$ as $n \to \infty$. Also, the sequence of dustributions $U(E_n)$ converges because the distribution function $F_n$ of the distribution $U(E_n)$ converges uniformly to some $F$ as $n \to \infty$, and obviously $F$ is also a distribution function. We call the distribution with the distribution function $F$ the uniform distribution on the Cantor set, denoted by $U(C)$. From Prop.~\ref{prop:continuity}, we can get $\delta_0 \sim_\mathrm{d} U(C)$.
\end{example}

\subsection{Compositionality}

Compositionality is a very important topic in model checking.
When a huge system is a composition of several small systems, we can work on these small systems to see whether their composition satisfy some property.
In this part we discuss the compositionality of our subdistribution bisimilarity.
This part also relies on Thm.~\ref{thm:metric} heavily.
We will see that two huge systems are subdistribution bisimilar if their composition components are subdistribution bisimilar, respectively, in our LMP settings.
We assume that all the LMPs in this part have the same action set $\mathcal{A}$.
First we introduce the definition of the composition for two LMPs:

\begin{definition}\label{def:composition}
Let $\mathcal{M}_i=(S_i,\Sigma_i,(\tau_a^i)_{a \in \mathcal{A}},\pi_i)$, $i=1,2$ be two LMPs. Their composition $\mathcal{M}_1 \mathbin{||} \mathcal{M}_2=(S,\Sigma,(\tau_a)_{a \in \mathcal{A}},\pi)$ is defined as follows:
\begin{itemize}
\item $(S,\Sigma)=(S_1 \times S_2,\sigma (\Sigma_1 \times \Sigma_2))$;
\item $\tau_a((s_1,s_2),\cdot)=\tau_a^1(s_1,\cdot) \times \tau_a^2(s_2,\cdot)$ for $(s_1,s_2) \in S$;
\item $\pi=\pi_1 \times \pi_2$.
\end{itemize}
\end{definition}

Then we show that composition preserves bisimilarity relation:

\begin{theorem}\label{thm:composition1}
$\mathcal{M}_1 \sim_\mathrm{d} \mathcal{M}_1'$ and $\mathcal{M}_2 \sim_\mathrm{d} \mathcal{M}_2'$ imply $\mathcal{M}_1 \mathbin{||} \mathcal{M}_2 \sim_\mathrm{d} \mathcal{M}_1' \mathbin{||} \mathcal{M}_2'$.
\end{theorem}

From Thm.~\ref{thm:composition1}, we can immediately know
that for any LMP $\mathcal{M}$, $\mathcal{M}_1 \sim_\mathrm{d} \mathcal{M}_1'$ implies
$\mathcal{M}_1 \mathbin{||} \mathcal{M} \sim_\mathrm{d} \mathcal{M}_1' \mathbin{||} \mathcal{M}$.
Actually Thm.~\ref{thm:composition1} is a special case of the following theorem,
by taking $\epsilon_1=\epsilon_2=0$:

\begin{theorem}\label{thm:composition2}
Given the discounting factor $c \in (0,1]$ and approximation $\epsilon_1,\epsilon_2 \in [0,1]$, $\mathcal{M}_1 \sim_{\epsilon_1}^c \mathcal{M}_1'$ and $\mathcal{M}_2 \sim_{\epsilon_2}^c \mathcal{M}_2'$ imply $\mathcal{M}_1 \mathbin{||} \mathcal{M}_2 \sim_{\epsilon_1+\epsilon_2-\epsilon_1 \epsilon_2}^c \mathcal{M}_1' \mathbin{||} \mathcal{M}_2'$.
\end{theorem}

Theorem~\ref{thm:composition2} bounds the distance between the composed LMPs.
This bound $\epsilon_1+\epsilon_2-\epsilon_1\epsilon_2$ can be approximated by $\epsilon_1+\epsilon_2$,
which is a linear function of $\epsilon_1$ and $\epsilon_2$.
Also we can see that composition with bisimilar LMPs does not make the distance of two LMPs larger,
so bisimulation is compositional in this sense.
Observe the bound in Thm.~\ref{thm:composition2} is tight.

\section{Conclusion}
In this paper we propose the definition of subdistribution bisimulation for LMPs, which is a bisimulation based on distributions rather than states and solve some basic problems on it. We compare it with previous bisimulations to show that it is a weaker bisimulation. Following a common way to study a bisimulation, we construct a logic and a metric both characterising our  subdistribution bisimulation.

There are several interesting directions for future works.
First, we plan to investigate an approximation scheme for our subdistribution bisimulation. Another direction is to deal with systems that are more complex than LMPs.
For example, we can add non-determinism choices for the same action, as the model in \cite{DTW12}.
In addition, we can add the internal action $\tau$ to the set of actions and investigate weak  bisimulations for  LMPs, and investigate the metric definition for continuous-time models \cite{FernsPP11}.

Last but not least, using coalgebras is a popular way
to describe bisimulation and simulation relations for probabilistic systems
(e.\,g.\@ \cite{VinkR99} and \cite{UrabeH14}),
and we expect that our distribution-based bisimulation for LMPs
and other more complex models will have a pretty coalgebraic description.

\subsection*{Acknowledgement}
This work has been supported by
by the National Natural Science Foundation of China (Grants 61532019, 61472473),
the CAS/SAFEA International Partnership Program for Creative Research Teams,
the Sino-German CDZ project CAP (GZ 1023).

\bibliography{ref}

\clearpage

\appendix

\section{Basics in Measure Theory} \label{A0}
Here we introduce some basic definitions and lemmas in measure theory for readers who are not familiar with this knowledge.

For a nonempty set $X$, we say $\mathcal{F} \subseteq 2^X$ is a $\sigma$-algebra on $X$, if
\begin{itemize}
\item $\varnothing \in \mathcal{F}$;
\item $A \in \mathcal{F}$ implies $A^\mathrm{c} \in \mathcal{F}$;
\item $A_1,A_2,\ldots,A_n,\ldots \in \mathcal{F}$ implies $\bigcup_n A_n \in \mathcal{F}$.
\end{itemize}
If $\mathcal{F}$ is a $\sigma$-algebra on $X$, we say $(X,\mathcal{F})$ is a measurable space. In addition, if $(X,\mathcal{T})$ is a topological space, we say $\mathcal{B}(X):=\sigma(\mathcal{T})$, the smallest $\sigma$-algebra containing $\mathcal{T}$, is the Borel $\sigma$-algebra on $X$.

\begin{definition}
Let $(X,\mathcal{F})$ be a measurable space. We say $\mu:\mathcal{F} \to [0,\infty]$ is a measure on $(X,\mathcal{F})$, if
\begin{itemize}
\item $\mu(\varnothing)=0$;
\item For $A_1,A_2,\ldots,A_n,\ldots \in \mathcal{F}$ satisfying $A_i \cap A_j =\varnothing$ whenever $i \ne j$, it holds that $\mu(\cup_n A_n)=\sum_n \mu(A_n)$.
\end{itemize}
If $\mu$ is a measure on $(X,\mathcal{F})$, we say $(X,\mathcal{F},\mu)$ is a measure space. If $\mu(X)<\infty$, we say $\mu$ is finite. If $\mu(X)=1$, we say $\mu$ is a probability measure or a distribution. If $\mu(X) \le 1$, we say $\mu$ is a subprobability measure or a subdistribution. We denote the set of all distribitions (subdistributions) on $(X,\mathcal{F})$ by $Dist(X,\mathcal{F})$ ($subDist(X,\mathcal{F})$), or simply $Dist(X)$ ($subDist(X)$). For $x \in X$, $\delta_x$ is the distribution satisfying $\delta_x(\{x\})=1$.
\end{definition}

\begin{definition}
Let $(X,\mathcal{F})$ and $(Y,\mathcal{G})$ be two measurable spaces.
We say a function $f:X \to Y$ is measurable,
denoted by $f:(X,\mathcal{F}) \to (Y,\mathcal{G})$,
if for any $A \in \mathcal{G}$, $f^{-1}(A) \in \mathcal{F}$.
If $f:(X,\mathcal{F}) \to (\mathbb{\bar R},\mathcal{B}(\mathbb{\bar R}))$ is measurable,
we simply say $f$ is Borel-measurable, where $\mathbb{\bar R}=\mathbb{R} \cup \{-\infty,\infty\}$.
\end{definition}

In particular, for $A \subseteq X$, we define the indicator function of $A$ as follows:
\begin{equation*}
\mathbb{I}_A(x):= \begin{cases}
1, & \text{if } x \in A; \\
0, & \text{if } x \in X \setminus A.
\end{cases}
\end{equation*}
On $(X,\mathcal{F})$, the function $\mathbb{I}_A$ is measurable if and only if $A \in \mathcal{F}$.

As labelled Markov processes may have continuous state spaces,
we extend the transition matrix to probability transition function.

\begin{definition}
Let $(X,\mathcal{F})$ be a measurable space. We say $\tau:X \times \mathcal{F} \to [0,1]$ is a probability (subprobability) transition function, if
\begin{itemize}
\item For any $A \in \mathcal{F}$, $\tau(\cdot,A)$ is Borel measurable;
\item For any $x \in X$, $\tau(x,\cdot)$ is a probability (subprobability) measure on $(X,\mathcal{F})$.
\end{itemize}
\end{definition}

\begin{definition}
Let $(X_i,\mathcal{F}_i)$, $i=1,2$ be two measurable spaces. We define their product space to be $(X_1 \times X_2, \sigma(\mathcal{F}_1 \times \mathcal{F}_2))$.
\end{definition}

In some places people simply use $\mathcal{F}_1 \times \mathcal{F}_2$ to denote the product $\sigma$-algebra. Let $A_1 \in \mathcal{F}_1$ and $A_2 \in \mathcal{F}_2$, and we call the set $A_1 \times A_2$ a measurable rectangle. Basically the product $\sigma$-algebra is the $\sigma$-algebra generated by the set of measurable rectangles.

Given two finite measure spaces $(X_i,\mathcal{F}_i,\mu_i)$, $i=1,2$, there exists a unique measure $\mu_1 \times \mu_2$ on $(X_1 \times X_2, \sigma(\mathcal{F}_1 \times \mathcal{F}_2))$, s.t. $(\mu_1 \times \mu_2)(A_1 \times A_2)=\mu_1(A_1) \mu_2(A_2)$.

We finally cite two famous lemmas in measure theory, which will be used in several places.

\begin{lemma}[Approximation with simple functions]\label{lemma:1}
We say a function $g:X \to Y$ is simple, if the range $g(X)$ is a finite set.
Let $f$ be a non-negative Borel measurable function on $(X,\mathcal{F})$.
Then there exists an increasing sequence of non-negative simple Borel measurable functions $\{f_i\}$,
s.\,t.\@ $f_i \uparrow f$ pointwise as $i \to \infty$.
Moreover, if $f$ is bounded, then there exists an increasing sequence of non-negative simple Borel measurable functions $\{f_i\}$,
s.\,t.\@ $f_i$ converges to $f$ uniformly as $i \to \infty$.
\end{lemma}

\begin{lemma}[Monotone convergence theorem]\label{lemma:2}
Let $\{f_i\}$ be a sequence of non-negative Borel measurable functions on $(X,\mathcal{F},\mu)$ and $f_i\uparrow f$ as $i \to \infty$. Then $f$ is Borel measurable, and
\begin{align*}
\lim_{i \to \infty}\int f_i \dif \mu = \int f \dif \mu.
\end{align*}
\end{lemma}

\section{Proof of Proposition~\ref{transitivity}} \label{A1}
\begin{proof}
Reflexivity and symmetry are trivial, and we only need to check transitivity.
Let $\mu_1,\mu_2,\mu_3 \in subDist(S)$ satisfy $\mu_1 \sim_{\mathrm{d}} \mu_2$ and $\mu_2 \sim_\mathrm{d} \mu_3$.
Then there exist two bisimulation relations $R_1$ and $R_2$, s.\,t.\@ $\mu_1 \mathrel{R_1} \mu_2$ and $\mu_2 \mathrel{R_2} \mu_2$.
Let $R:=\{(\mu,\mu'):\exists \nu.~(\mu \mathrel{R_1} \nu \wedge \nu \mathrel{R_2} \mu')\}$.
Then we have $\mu_1 \mathrel{R} \mu_3 $.
It suffices to show that $R$ is a bisimulation relation.
We assume that $\mu \mathrel{R} \mu'$.
Then there exists $\nu \in subDist(S)$, s.\,t.\@ $\mu \mathrel{R_1} \nu$ and $\nu \mathrel{R_2} \mu'$.
Because $R_1$ and $R_2$ are bisimulation relations, we have $\mu(S)=\nu(S)=\mu'(S)$.
For any $a \in \mathcal{A}$, let $\mu \xrightarrow{a} \tilde{\mu}$, $\nu \xrightarrow{a} \tilde{\nu}$, and $\mu' \xrightarrow{a} \tilde{\mu}'$.
Because $R_1$ and $R_2$ are bisimulation relations,
we have $\tilde{\mu} \mathrel{R_1} \tilde{\nu}$ and $\tilde{\nu} \mathrel{R_2} \tilde{\mu}'$.
From the definition of $R$, we have $\tilde{\mu} \mathrel{R} \tilde{\mu}'$.
Therefore, $R$ is a bisimulation relation.
\end{proof}

\section{Proof of Theorem~\ref{biimply}}
\begin{proof}
We assume $s \sim_{\mathrm{s}} t$.
Then there exists a state bisimulation relation $R \subseteq S \times S$, s.\,t.\@ $s \mathrel{R} t$.
We define the lifted relation $\tilde{R} \subseteq subDist(S) \times subDist(S)$ as follows:
$\mu \mathrel{\tilde{R}} \nu$ if and only if for all $R$-closed sets $A \in \Sigma$, $\mu(A)=\nu(A)$.
It is easy to check $\delta_s \mathrel{\tilde{R}} \delta_t$, so it remains to prove that $\tilde{R}$ is a bisimulation relation.
We assume $\mu \mathrel{\tilde{R}} \nu$, $\mu \xrightarrow{a} \mu'$ and $\nu \xrightarrow{a} \nu'$.
Obviously we have $\mu(S)=\nu(S)$ since $S$ is $R$-closed.
We need to show that, for any $R$-closed set $A \in \Sigma$,
$\int \tau_a(s,A) \mu(\dif s)=\int \tau_a(s,A) \nu(\dif s)$.

We first assume $\tau_a(\cdot,A)$ is a simple function, i.\,e.\@ $\tau_a(\cdot,A)=\sum_{i=1}^n a_i \mathbb{I}_{A_i}$,
where $A_i \in \Sigma$, $A_i \cap A_j =\varnothing$ and $a_i \ne a_j$ whenever $i \ne j$.
We notice that $\tau_a(\cdot,A)$ is constant on every $R$-equivalence class (see \eqref{eq:statebi}), so every $A_i$ is $R$-closed.
Then
\begin{align*}
\int \tau_a(s,A) \mu(\dif s)&=\sum_{i=1}^n a_i\mu(A_i)=\sum_{i=1}^n a_i\nu(A_i)\\
&=\int \tau_a(s,A) \nu(\dif s).
\end{align*}
For general $\tau_a(\cdot,A)$, we let
\begin{align*}
f_n=\sum_{i=1}^{2^n} \frac{i}{2^n} \mathbb{I}_{\{\frac{k}{2^n} \le \tau_a(\cdot,A) <\frac{k+1}{2^n}\}}.
\end{align*}
Then every $f_n$ is a simple Borel-measurable function and attains a constant on every $R$-equivalence class, so we have
$\int f_n \dif\mu=\int f_n \dif \nu$.
It is easy to check that $f_n$ is increasing and $0 \le \tau_a(\cdot,A)-f_n \le 2^{-n}$, so we have $f_n \uparrow \tau_a(\cdot,A)$.
From the monotone convergence theorem, as $n \to \infty$, we have
$ \int f_n \dif\mu \to \int \tau_a(s,A) \mu(\dif s)$ and $\int f_n \dif \nu \to \int \tau_a(s,A) \nu(\dif s)$,
which indicates that
\begin{align*}
\int \tau_a(s,A) \mu(\dif s)=\int \tau_a(s,A) \nu(\dif s).
\end{align*}
Therefore $\mu' \mathrel{\tilde{R}} \nu'$, i.\,e.\@ $\tilde{R}$ is a bisimulation relation.

For the other direction, one counterexample is just Exa.~\ref{subbiex}. We have $\delta_{s_0} \sim_{\mathrm{d}} \delta_{t_0}$, but $s_0 \sim_{\mathrm{s}} t_0$ does not hold. This is because the behaviour of $s_1$ can not be simulated by $t_1$ or $t_2$.
\end{proof}

\section{Proof of Theorem~\ref{logicalcharacterization}} \label{A3}
We divide the proof into two parts, soundness and completeness.

\begin{lemma}[Soundness]\label{soundness}
If $\mu \sim_{\mathrm{d}} \nu$, then $\mu$ and $\nu$ satisfy the same formulae in $\mathcal{L}_0$.
\end{lemma}
\begin{proof}
We assume $\mu \sim \nu$. Then there exists a bisimulation relation $R$, s.\,t.\@ $\mu \mathrel{R} \nu$. We show that $\mu \mathrel{R} \nu$ implies that $\mu$ and $\nu$ satisfy the same formulae by structural induction on $\mathcal{L}_0$.
\begin{itemize}
\item It is obvious that for all $(\mu,\nu) \in R$, $\mu \models \mathbb{T}$ and $\nu \models \mathbb{T}$.
\item If for all $(\mu,\nu) \in R$, $\mu \models \varphi_i$ if and only if $\nu \models \varphi_i$, $i=1,2$, then obviously $\mu \models \varphi_1\wedge\varphi_2$ if and only if $\nu \models \varphi_1\wedge\varphi_2$.
\item If for all $(\mu,\nu) \in R$, $\mu \models \varphi$ if and only if $\nu \models \varphi$, then for any $\mu \mathrel{R} \nu$,
\begin{align*}
\mu \models \langle a \rangle_q \varphi & \quad \text{iff (by Definition of $\langle a \rangle_q$)}\\
\mu' \models \varphi \text{ and } \mu'(S) \ge q & \quad \text{iff (by I. H and $\mu'(S) = \nu'(S)$)}\\
\nu' \models \varphi \text{ and } \nu'(S) \ge q & \quad \text{iff (by Definition of $\langle a \rangle_q$)}\\
\nu \models \langle a \rangle_q \varphi\makebox[0pt][l]{,}
\end{align*}
where $\mu \xrightarrow{a} \mu'$ and $\nu \xrightarrow{a} \nu'$.
\item For all $(\mu,\nu) \in R$, we have $\mu(S)=\nu(S)$, so $\mu \models \langle \varepsilon \rangle_q$ if and only if $\nu \models \langle \varepsilon \rangle_q$.
\end{itemize}
Therefore, $\mu$ and $\nu$ satisfy the same formulae in $\mathcal{L}_0$.
\end{proof}

\begin{lemma}[Completeness]\label{completeness}
If $\mu,\nu \in subDist(S)$ satisfy the same formulae in $\mathcal{L}_0$, then $\mu \sim_{\mathrm{d}} \nu$.
\end{lemma}
\begin{proof}
It suffices to show that the relation
\begin{align*}
R:=\{(\mu,\nu):\mu \text{ and }\nu \text{ satisfy the same formulae in }\mathcal{L}_0\}
\end{align*}
is a bisimulation relation.
We assume $\mu \mathrel{R} \nu$, $\mu \xrightarrow{a} \mu'$ and $\nu \xrightarrow{a} \nu'$.
First we show $\mu(S)=\nu(S)$.
Since $\mu \mathrel{R} \nu$, $\mu \models \langle \varepsilon \rangle_q$ if and only if $\nu \models \langle \varepsilon \rangle_q$,
i.\,e.\@ for any $q \in \mathbb{Q}$, $\mu(S) \ge q$ if and only if $\nu(S) \ge q$,
which implies $\mu(S)=\nu(S)$.

Now assume that $\mu' \mathrel{R} \nu'$ does not hold,
so w.\,l.\,o.\,g.\@ there exists a formula $\varphi$, s.\,t.\@ $\mu' \models \varphi$ and $\nu' \not \models \varphi$.
Then we consider the formula $\psi=\langle a \rangle_q \varphi$, where $q=\min\{\mu'(S),\nu'(S)\}$.
$\mu' \models \varphi$ implies $\mu \models \psi$, so $\nu \models \psi$.
Then we must have $\nu' \models \varphi$. Contradiction!
Therefore, $\mu' \mathrel{R} \nu'$, and $R$ is indeed a bisimulation relation.
\end{proof}

\section{Proof of Proposition~\ref{pro:extendedlogic}}
\begin{proof}
From Lemmas~\ref{soundness} and \ref{completeness}, it suffices to show by structural induction that, $\mu \sim_{\mathrm{d}} \nu$ implies that they satisfy the same formulae in $\mathcal{L}_{\mathrm{Can}}$, $\mathcal{L}_{\Delta}$, $\mathcal{L}_{\neg}$ and $\mathcal{L}_{\wedge}$, and we only need to check $\mathrm{Can}(a)$, $\Delta_a$, $\neg \varphi$, and $\bigwedge_{i \in \mathbb{N}} \varphi_i$. We assume $\mu \xrightarrow{a} \mu'$ and $\nu \xrightarrow{a} \nu'$. Since $\mu \sim_{\mathrm{d}} \nu$, we have $\mu' \sim_{\mathrm{d}} \nu'$, and naturally $\mu'(S)=\nu'(S)$.
\begin{itemize}
\item $\mu \models \mathrm{Can}(a)$ iff $\mu'(S)>0$ iff $\nu'(S)>0$ iff $\nu \models \mathrm{Can}(a)$.
\item $\mu \models \Delta_a$ iff $\mu'(S)=0$ iff $\nu'(S)=0$ iff $\nu \models \Delta_a$.
\item The proof of negation and countable conjunction is obvious.
\end{itemize}
\end{proof}

\section{Proof of Proposition~\ref{eventimply}}
\begin{proof}
For any $\varphi \in \mathcal{L}_0$, we turn it into CNF $\bigwedge_{i=1}^m \varphi_i$, where $\varphi_i=\langle a_{i,1} \rangle_{q_{i,1}} \cdots \langle a_{i,n_i} \rangle_{q_{i,n_i}} \mathbb{T}$.
Consider $\phi=\bigwedge_{i=1}^m \phi_i \in \mathcal{L}$, where $\phi_i=\langle a_{i,1} \rangle_{q_{i,1}}^\mathrm{st} \cdots  \langle a_{i,n_i} \rangle_{q_{i,n_i}}^\mathrm{st} \mathbb{T}$. Then we have
\begin{align*}
\delta_s \models \varphi \enspace\text{ iff} \\
\delta_s \models \varphi_i, \enspace i=1,\ldots,m \enspace\text{ iff} \\
s \models \phi_i, \enspace i=1,\ldots,m \enspace\text{ iff} \\
t \models \phi_i, \enspace i=1,\ldots,m \enspace\text{ iff} \\
\delta_t \models \varphi_i, \enspace i=1,\ldots,m \enspace\text{ iff} \\
\delta_t \models \varphi.\enspace\enspace\enspace
\end{align*}
\end{proof}

\section{Proof of Theorem~\ref{thm:metric}}
\begin{proof}
If $\mu \sim \nu$, then from the definition of subdistribution bisimulation, one easily proves $d^c(\mu,\nu)=0$
by induction on the length of the word $w$ in Def.~\ref{defi:metric}.
%(Just do induction on the length of the word $w$
%to prove that for any $w \in \mathcal{A}^*$, $\mu'(S)=\nu'(S)$,
%where $\mu \xrightarrow{w} \mu'$ and $\nu \xrightarrow{w} \nu'$.)

For the other direction, we only need to show that the relation $R:=\{(\mu,\nu):d^c(\mu,\nu)=0\}$ is a bisimulation.
We assume that $\mu \mathrel{R} \nu$.
Because $d^c(\mu,\nu)=0$, we have $\mu(S)=\nu(S)$.
Then for any $a \in \mathcal{A}$, let $\mu \xrightarrow{a} \mu'$ and $\nu \xrightarrow{a} \nu'$. We have
\begin{align*}
d^c(\mu',\nu')&=\sup_{w \in \mathcal{A}^*,\mu' \xrightarrow{w} \mu'',\nu' \xrightarrow{w} \nu''} c^{|w|}|\mu''(S)-\nu''(S)| \\
&=\sup_{w \in a\mathcal{A}^*,\mu \xrightarrow{w} \mu'',\nu \xrightarrow{w} \nu''}c^{|w|}|\mu''(S)-\nu''(S)| \\
& \le \sup_{w \in \mathcal{A}^*,\mu \xrightarrow{w} \mu'',\nu \xrightarrow{w} \nu''}c^{|w|}|\mu''(S)-\nu''(S)|
=d^c(\mu,\nu)=0,
\end{align*}
where $a\mathcal{A}^*:=\{aw:w \in \mathcal{A}^*\}$. Therefore, $d^c(\mu',\nu')=0$ and we have $\mu' \mathrel{R} \nu'$, which implies that $R$ is a bisimulation relation.
\end{proof}

\section{Proof of Proposition~\ref{prop:approx}} \label{A5}
\begin{proof}
First we assume $\mu \sim_{\epsilon'}^c \nu$. It is easy to see that $\{\sim_{\epsilon}^c\}$ is an approximating bisimulation relation,
so $\mu \approx_{\epsilon'}^c \nu$.

Now we assume $\mu \approx_{\epsilon'}^c \nu$. Then there exists an approximating bisimulation relation $\{R_\epsilon^c\}$, s.\,t.\@ $\mu \mathrel{R_{\epsilon'}^c} \nu$. It suffices to show $d^c(\mu,\nu) \le \epsilon'$, i.\,e.\@ for any $w \in \mathcal{A}^*$, $|\mu'(S)-\nu'(S)| \le c^{-|w|}\epsilon'$, where $\mu \xrightarrow{w} \mu'$ and $\nu \xrightarrow{w} \nu'$. From the definition, we can see that, for $w \in \mathcal{A}^*$ with $|w|=n$, we have $\mu' \mathrel{R_{c^{-n}\epsilon'}^c} \nu'$, and thus $|\mu'(S)-\nu'(S)| \le c^{-|w|}\epsilon'$.
\end{proof}

\section{Proof of Proposition~\ref{prop:logicalmetric}} \label{A6}
\begin{proof}
First we prove $d^c \ge d_{\mathrm{l}}^c$. It suffices to show by structural induction that, for any $\varphi \in \mathcal{L}_\mathcal{M}^c$, $d^c(\mu,\nu) \ge |\varphi(\mu)-\varphi(\nu)|$.
\begin{itemize}
\item $\varphi=\mathbf{1}$. Then $|\varphi(\mu)-\varphi(\nu)|=|\mu(S)-\nu(S)| \le d^c(\mu,\nu)$.
\item $\varphi=\varphi' \oplus p$. Without loss of generality, we assume $\varphi'(\mu) \ge \varphi'(\nu)$. Then $\varphi(\mu) \ge \varphi(\nu)$, and
$|\varphi(\mu)-\varphi(\nu)|=\min\{\varphi'(\mu)+p,1\}-\min\{\varphi'(\nu)+p,1\} \le \varphi'(\mu)-\varphi'(\nu) \le d^c(\mu,\nu)$.
\item $\varphi=\neg \varphi'$. Then
$|\varphi(\mu)-\varphi(\nu)|=|1-\varphi'(\mu)-1+\varphi'(\nu)| = \varphi'(\mu)-\varphi'(\nu) \le d^c(\mu,\nu)$.
\item $\varphi=\bigwedge_{i \in I} \varphi_i$. Without loss of generality, we assume $\varphi(\mu) \ge \varphi(\nu)$. For any $\epsilon>0$, there exist $j \in I$, s.\,t.\@ $\varphi_j(\nu) \le \varphi(\nu)+\epsilon$. Then
$|\varphi(\mu)-\varphi(\nu)|=\varphi(\mu)-\varphi(\nu) \le \varphi_j(\mu)-\varphi_j(\nu) +\epsilon \le d^c(\mu,\nu)+\epsilon$,
and we have $|\varphi(\mu)-\varphi(\nu)| \le d^c(\mu,\nu)$ since $\epsilon>0$ is arbitrary.
\item $\varphi=\langle a\rangle^c \varphi'$. Let $\mu \xrightarrow{a} \mu'$ and $\nu \xrightarrow{a} \nu'$. Then
$|\varphi(\mu)-\varphi(\nu)|=c|\varphi'(\mu')-\varphi'(\nu')| \le c\cdot d^c(\mu',\nu') \le d^c(\mu,\nu)$.
\end{itemize}

Then we show $d^c \le d_{\mathrm{l}}^c$. We define a sub-logic of $\mathcal{L}_\mathcal{M}^c$ as follows:
\begin{align*}
\mathcal{L}'::=\mathbf{1}~|~\langle a \rangle^c \varphi.
\end{align*}
Then
\begin{align*}
d^c(\mu,\nu)&=\sup_{w \in \mathcal{A}^*,\mu \xrightarrow{w} \mu',\nu \xrightarrow{w} \nu'} c^{|w|}|\mu'(S)-\nu'(S)|\\
&=\sup_{w \in \mathcal{A}^*} |\langle w \rangle^c \mathbf{1}(\mu)-\langle w \rangle^c \mathbf{1}(\nu)| \\
&=\sup_{\varphi \in \mathcal{L}'} |\varphi(\mu)-\varphi(\mu)| \le \sup_{\varphi \in \mathcal{L}_\mathcal{M}^c} |\varphi(\mu)-\varphi(\mu)| \\
&=d_{\mathrm{l}}^c(\mu,\nu),
\end{align*}
where for $w=a_1 \cdots a_n$, let $\langle w \rangle^c \mathbf{1}$ be an abbreviation for $\langle a_n \rangle^c \cdots \langle a_1 \rangle^c \mathbf{1}$.
\end{proof}

\section{Proof of Lemma~\ref{lemma:linear}} \label{A7}
\begin{proof}
If $w=\varepsilon$, the proof is trivial.

Now we assume $w=a \in \mathcal{A}$. We prove that, for any measure $\mu,\nu$ on $(X,\mathcal{F})$ and any non-negative Borel measurable function $f$,
\begin{align}\label{eq:easy1}
\int f \dif \mu + \int f \dif \nu = \int f \dif~ (\mu+\nu).
\end{align}
If $f$ is a simple function, i.\,e.\@ $f=\sum_{i=1}^m b_i \mathbb{I}_{B_i}$, where $b_i \ge 0$ and $B_i \in \mathcal{F}$, then
\begin{align*}
\int f \dif \mu + \int f \dif \nu &= \sum_{i=1}^m b_i\mu(B_i)+\sum_{i=1}^m b_i\nu(B_i)\\
&=\sum_{i=1}^m b_i(\mu(B_i)+\nu(B_i))  \\
&=\sum_{i=1}^m b_i(\mu+\nu)(B_i)=\int f \dif~ (\mu+\nu).
\end{align*}
Now we assume that $f$ is a non-negative Borel measurable function. Then there exists a sequence of non-negative simple measurable functions $\{f_i\}$, s.\,t.\@ $f_i \uparrow f$. From the monotone convergence theorem, as $i \to \infty$, we have
\begin{align*}
\int f_i \dif~ (\mu+\nu) \to \int f \dif~ (\mu+\nu),
\end{align*}
and
\begin{align*}
\int f_i \dif~ (\mu+\nu) =\int f_i \dif \mu + \int f_i \dif \nu \to \int f \dif \mu + \int f \dif \nu,
\end{align*}
which imply \eqref{eq:easy1}. In addition, following the similar way it is easy to prove that, for any $a \ge 0$,
\begin{align}\label{eq:easy2}
\int f \dif ~(a\mu) = a \int f \dif \mu.
\end{align}
Then directly from \eqref{eq:easy1} and \eqref{eq:easy2}, we can see that the relation $\xrightarrow{a}$ is linear.

For longer $w$, we can prove it by doing induction on its length and using the fact that if $\mu\xrightarrow{w}\mu'$ and $\mu'\xrightarrow{w'}\mu''$, then $\mu\xrightarrow{ww'}\mu''$.

For $\sigma$-linearity, the proof is quite similar , except for proving the fact that for any measure $\mu_1,\mu_2,\ldots$ and on $(X,\mathcal{F})$ and any non-negative Borel measurable function $f$,
\begin{align*}%\label{eq:easy3}
\sum_{n=1}^\infty \int f \dif \mu_n  = \int f \dif~ \left(\sum_{n=1}^\infty \mu_n\right).
\end{align*}
If $f$ is a simple function, i.\,e.\@ $f=\sum_{i=1}^m b_i \mathbb{I}_{B_i}$, where $b_i \ge 0$ and $B_i \in \mathcal{F}$, then
\begin{align*}
\sum_{n=1}^\infty \int f \dif \mu_n &=\sum_{n=1}^\infty \sum_{i=1}^m b_i\mu_n(B_i)=\sum_{i=1}^m b_i\sum_{n=1}^\infty \mu_n(B_i) \\
&=\sum_{i=1}^m b_i\left(\sum_{n=1}^\infty\mu_n\right)(B_i)=\int f \dif~ \left(\sum_{n=1}^\infty \mu_n\right).
\end{align*}
Now we assume that $f$ is a non-negative Borel measurable function. Then there exists a sequence of non-negative simple measurable functions $\{f_i\}$, s.\,t.\@ $f_i \uparrow f$. From the monotone convergence theorem, we first have
\begin{align*}
\int f_i \dif~ \left(\sum_{n=1}^\infty\mu_n\right) \to \int f \dif~ \left(\sum_{n=1}^\infty\mu_n\right), \enspace i \to \infty,
\end{align*}
and
\begin{align*}
\int f_i \dif \mu_n \to \int f \dif \mu_n, \enspace i \to \infty.
\end{align*}
Then by again applying the monotone convergence theorem, we have
\begin{align*}
\int f_i \dif~\left(\sum_{n=1}^\infty\mu_n\right) &=\sum_{n=1}^\infty \int f_i \dif \mu_n=\int \dif\#\int f_i \dif \mu_n \\
&\to \int \dif\#\int f \dif \mu_n  =\sum_{n=1}^\infty\int f \dif \mu_n
\end{align*}
as $i \to \infty$, where $\#$ is the counting measure on $(\mathbb{N},2^{\mathbb{N}})$, i.\,e.\@ $\#(A)$ is the number of elements in $A$.
\end{proof}

\section{Proof of Lemma~\ref{lemma:continuous}} \label{A8}
To prove Lemma~\ref{lemma:continuous}, we need the following lemma.

\begin{lemma}\label{lemma:cenvergence}
Let $L:\mathbb{N} \times \mathbb{N} \to \mathbb{R}$ be a function. If the following conditions hold:
\begin{itemize}
\item $\lim_{m \to \infty} L(m,n)$ and $\lim_{n \to \infty} L(m,n)$ exist for all $n$ and $m$, respectively;
\item for any $\epsilon>0$, there exists $M \in \mathbb{N}$, s.\,t.\@ for any $m_1,m_2>M$ and $n \in \mathbb{N}$, $|L(m_1,n)-L(m_2,n)| \le \epsilon$
($L(\cdot,n)$ converges uniformly in $n$),
\end{itemize}
then the repeated limits $\lim_{m \to \infty} \lim_{n \to \infty} L(m,n)$ and $\lim_{n \to \infty} \lim_{m \to \infty} L(m,n)$ exist, and
\begin{align*}
\lim_{m \to \infty} \lim_{n \to \infty} L(m,n)=\lim_{n \to \infty} \lim_{m \to \infty} L(m,n).
\end{align*}
\end{lemma}
\begin{proof}
Let $g(m)=\lim_{n \to \infty} L(m,n)$ and $h(n)=\lim_{m \to \infty} L(m,n)$. For any $\epsilon>0$, there exists $M \in \mathbb{N}$, s.\,t.\@ for any $m_1,m_2>M$ and any $n \in \mathbb{N}$, $|L(m_1,n)-L(m_2,n)| \le \epsilon$.
Taking the limit $n \to \infty$, we get $|g(m_1)-g(m_2)|\le \epsilon$.
From the Cauchy convergence criterion, we know that
$a=\lim_{m \to \infty} \lim_{n \to \infty} L(m,n)$
exists.

There exists $M' \in \mathbb{N}$, s.\,t.\@ for any $m>M'$, $|h(n)-L(m,n)|<\epsilon/3$ and $|g(m)-a|<\epsilon/3$. Since we have $g(m)=\lim_{n \to \infty} L(m,n)$, then there exists $N \in  \mathbb{N}$, s.\,t.\@ for any $n>N$, $|L(M',n)-g(M')|<\epsilon/3$. Now for any $n>N$,
$|h(n)-a| \le |h(n)-L(M',n)|+|L(M',n)-g(M')|+|g(M')-a| <\epsilon$,
which indicates that $\lim_{n \to \infty} \lim_{m \to \infty} L(m,n)=a$.
\end{proof}

Then we start to prove Lemma~\ref{lemma:continuous}.

\begin{proof}
It suffices to show that, for any non-negative bounded Borel measurable function $f$, if $\mu_n \to \mu$ as $n \to \infty$, then
\begin{align*}
\lim_{n \to \infty}\int f \dif \mu_n=\int f \dif \mu.
\end{align*}
If $f$ is a simple function, i.\,e.\@ $f=\sum_{i=1}^m b_i \mathbb{I}_{B_i}$, where $b_i \ge 0$ and $B_i \in \mathcal{F}$, then
\begin{align*}
\int f \dif \mu_n=\sum_{i=1}^m b_i \mu_n(B_i) \to \sum_{i=1}^m b_i \mu(B_i)=\int f \dif \mu
\end{align*}
as $n \to \infty$. Now we assume that $f$ is a Borel-measurable function. Then there exists an increasing sequence $\{f_i\}$, s.\,t.\@ $f_i$ converges to $f$ uniformly, i.\,e.\@ for any $\epsilon>0$, there exists $M>0$, s.\,t.\@ for any $m_1,m_2>M$, $|f_{m_1}-f_{m_2}|<\epsilon$. Let
\begin{align*}
L(m,n)=\int f_m \dif \mu_n.
\end{align*}
First, $\lim_{n \to \infty} L(m,n)$ exists from the proof above. Then, $\lim_{m \to \infty} L(m,n)$ exists from the monotone convergence theorem. Finally, for any $m_1,m_2>M$ and any $n \in \mathbb{N}$,
\begin{align*}
|L(m_1,n)-L(m_2,n)|&=\left|\int f_{m_1} \dif \mu_n-\int f_{m_2} \dif \mu_n\right| \\
&=\left|\int (f_{m_1}-f_{m_2}) \dif \mu_n \right|\\
&\le \int |f_{m_1}-f_{m_2}| \dif \mu_n < \epsilon.
\end{align*}
From Lemma~\ref{lemma:cenvergence}, we know $\lim_{m \to \infty} \lim_{n \to \infty} L(m,n)=\lim_{n \to \infty} \lim_{m \to \infty} L(m,n)$. From the monotone convergence theorem, we have
\begin{align*}
\lim_{m \to \infty} \lim_{n \to \infty} L(m,n)=\lim_{m \to \infty} \int f_m \dif \mu=\int f \dif \mu,
\end{align*}
and
\begin{align*}
\lim_{n \to \infty} \lim_{m \to \infty} L(m,n)=\lim_{n \to \infty}\int f \dif \mu_n.
\end{align*}
Therefore, we have
\begin{align*}
\lim_{n \to \infty}\int f \dif \mu_n=\int f \dif \mu.
\end{align*}
\end{proof}

\iffalse
\section{Proof of Theorem~\ref{thm:composition1}} \label{A9}
\begin{proof}
Let $\mathcal{M}_i=(S_i,\Sigma_i,(\tau_a^i)_{a \in \mathcal{A}},\pi_i)$ and $\mathcal{M}_i'=(S_i',\Sigma_i',(\tau_a^i{}')_{a \in \mathcal{A}},\pi_i')$, $i=1,2$.
From Thm.~\ref{thm:metric} we have $d^1(\pi_1,\pi_1')=d^1(\pi_2,\pi_2')=0$,
i.\,e.\@ for any word $w \in \mathcal{A}^*$, $\mu_1(S_1)=\mu_1'(S_1')$ and $\mu_2(S_2)=\mu_2'(S_2')$,
where $\pi_i \stackrel{w}{\to} \mu_i$ and $\pi_i' \stackrel{w}{\to} \mu_i'$, $i=1,2$.
Then $(\mu_1 \times \mu_2)(S_1 \times S_2)=\mu_1(S_1) \times \mu_2(S_2)=\mu_1'(S_1') \times \mu_2'(S_2')=(\mu_1' \times \mu_2')(S_1' \times S_2')$.
It is easy to see $\pi_1 \times \pi_2  \stackrel{w}{\to} \mu_1 \times \mu_2$ and $\pi_1' \times \pi_2' \stackrel{w}{\to} \mu_1' \times \mu_2'$ in $\mathcal{M}_1 \mathbin{||} \mathcal{M}_2$ and $\mathcal{M}_1' \mathbin{||} \mathcal{M}_2'$, respectively.
Then we can see $d^1(\pi_1 \times \pi_2,\pi_1' \times \pi_2')=0$ since $w$ is arbitaray.
We again use Thm.~\ref{thm:metric} and get $\mathcal{M}_1 \mathbin{||} \mathcal{M}_2 \sim_\mathrm{d} \mathcal{M}_1' \mathbin{||} \mathcal{M}_2'$.
\end{proof}
\fi

\section{Proof of Theorem~\ref{thm:composition2}}
\begin{proof}
Let $\mathcal{M}_i=(S_i,\Sigma_i,(\tau_a^i)_{a \in \mathcal{A}},\pi_i)$ and $\mathcal{M}_i'=(S_i',\Sigma_i',(\tau_a^i{}')_{a \in \mathcal{A}},\pi_i')$, $i=1,2$.
From the definition, for any $w \in \mathcal{A}^*$ with $|w|=n$, we have
\begin{align}\label{4.3}
|\mu_1(S_1)-\mu_1'(S_1')| \le \frac {\epsilon_1} {c^n}, \enspace |\mu_2(S_2)-\mu_2'(S_2')| \le \frac {\epsilon_2} {c^n},
\end{align}
where $\pi_i \stackrel{w}{\to} \mu_i$ and $\pi_i' \stackrel{w}{\to} \mu_i'$, $i=1,2$.
Without loss of generality, we assume $\mu_1(S_1)\mu_2(S_2) \ge \mu_1'(S_1')\mu_2'(S_2')$.
Now we consider $\mu_1(S_1)\mu_2(S_2)-\mu_1'(S_1')\mu_2'(S_2')$.
\begin{itemize}
\item $\mu_1(S_1) \ge \dfrac {\epsilon_1} {c^n}$ and $\mu_2(S_2) \ge \dfrac {\epsilon_2} {c^n}$.
Then given Inequality~\eqref{4.3}, we have
\begin{align*}
&\mu_1(S_1)\mu_2(S_2)-\mu_1'(S_1')\mu_2'(S_2') \\
\le{}& \mu_1(S_1)\mu_2(S_2)-(\mu_1(S_1)- \frac {\epsilon_1} {c^n})(\mu_2(S_2)-\frac {\epsilon_2} {c^n}) \\
\le{}& 1 \cdot 1 -(1-\frac{\epsilon_1} {c^n})(1-\frac {\epsilon_2} {c^n}) \\
={}& \frac{\epsilon_1} {c^n}+\frac {\epsilon_2}{c^n}-\frac{\epsilon_1\epsilon_2}{c^{2n}}.
\end{align*}
\item $\mu_1(S_1) \ge \dfrac {\epsilon_1} {c^n}$ and $\mu_2(S_2) < \dfrac {\epsilon_2} {c^n}$. Then we have
\begin{align*}
&\mu_1(S_1)\mu_2(S_2)-\mu_1'(S_1')\mu_2'(S_2') \\
\le{}& \mu_1(S_1)\mu_2(S_2)-(\mu_1(S_1)- \frac {\epsilon_1} {c^n})\mu_2(S_2) \\
\le{}& 1 \cdot 1 -(1-\frac{\epsilon_1} {c^n})\cdot 1 =\frac{\epsilon_1} {c^n}.
\end{align*}
\item $\mu_1(S_1) < \dfrac {\epsilon_1} {c^n}$ and $\mu_2(S_2) \ge \dfrac {\epsilon_2} {c^n}$. Then we have
\begin{align*}
&\mu_1(S_1)\mu_2(S_2)-\mu_1'(S_1')\mu_2'(S_2') \\
\le{}& \mu_1(S_1)\mu_2(S_2)-\mu_1(S_1)(\mu_2(S_2)- \dfrac {\epsilon_2} {c^n})\\
\le{}& 1 \cdot 1 -1\cdot(1 - \dfrac {\epsilon_2} {c^n})=\frac{\epsilon_2} {c^n}.
\end{align*}
\item $\mu_1(S_1) < \dfrac {\epsilon_1} {c^n}$ and $\mu_2(S_2) < \dfrac {\epsilon_2} {c^n}$. Then we have
\begin{align*}
&\mu_1(S_1)\mu_2(S_2)-\mu_1'(S_1')\mu_2'(S_2') \\
\le{}& \mu_1(S_1)\mu_2(S_2) \\
\le{}& \min\{1,\dfrac {\epsilon_1} {c^n}\}\min\{1,\dfrac {\epsilon_2} {c^n}\}.
\end{align*}
\end{itemize}
Then it is easy to check
\begin{align*}
&\mu_1(S_1)\mu_2(S_2)-\mu_1'(S_1')\mu_2'(S_2') \\
\le{}& \max\left\{\frac{\epsilon_1} {c^n}+\frac {\epsilon_2}{c^n}-\frac{\epsilon_1\epsilon_2}{c^{2n}},\frac{\epsilon_1} {c^n},\frac{\epsilon_2} {c^n}, \min\{1,\dfrac {\epsilon_1} {c^n}\}\min\{1,\dfrac {\epsilon_2} {c^n}\}\right\} \\
\le{}& \frac{\epsilon_1} {c^n}+\frac {\epsilon_2}{c^n}-\frac{\epsilon_1\epsilon_2}{c^n}.
\end{align*}
Therefore,
\begin{align*}
d^c(\mathcal{M}_1 \mathbin{||} \mathcal{M}_2,\mathcal{M}_1' \mathbin{||} \mathcal{M}_2') &\le c^n \left(\dfrac{\epsilon_1} {c^n}+\dfrac {\epsilon_2}{c^n}-\dfrac{\epsilon_1\epsilon_2}{c^n} \right)\\
&=\epsilon_1+\epsilon_2-\epsilon_1 \epsilon_2.
\end{align*}
since $w$ is arbitrary.
Because $\epsilon_1,\epsilon_2 \in [0,1]$, ${\epsilon_1} + {\epsilon_2}-{\epsilon_1\epsilon_2} \in [0,1]$, and we have $\mathcal{M}_1 \mathbin{||} \mathcal{M}_2 \sim_{\epsilon_1+\epsilon_2-\epsilon_1 \epsilon_2}^c \mathcal{M}_1' \mathbin{||} \mathcal{M}_2'$.
\end{proof}

\end{document}